\documentclass[sigconf,prologue,usenames,dvipsnames]{acmart}
\hypersetup{
  breaklinks=true,   
}
\usepackage[utf8]{inputenc}

\usepackage{filecontents}  %

\usepackage[nohyperlinks]{acronym}
\acrodef{PC}{program counter}
\acrodef{MOP}{microarchitectural optimization}
\acrodef{CFL}{control-flow leakage}
\acrodef{CFG}{control-flow graph}
\acrodef{CF}{control flow}

\usepackage{tabularx} %
\usepackage{booktabs} %

\usepackage{makecell}

\usepackage{array}
\newcolumntype{L}[1]{>{\raggedright\let\newline\\\arraybackslash\hspace{0pt}}p{#1}}
\newcolumntype{C}[1]{>{\centering\let\newline\\\arraybackslash\hspace{0pt}}p{#1}}
\newcolumntype{R}[1]{>{\raggedleft\let\newline\\\arraybackslash\hspace{0pt}}p{#1}}

\usepackage{graphicx}

\usepackage{caption}
\usepackage{subcaption}

\usepackage{amsthm}
\usepackage{amsfonts}
\usepackage{amsmath}
\usepackage[capitalize,noabbrev,nameinlink]{cleveref}
\crefname{lstlisting}{Listing}{Listings}
\Crefname{lstlisting}{Listing}{Listings}
\usepackage{stmaryrd}
\usepackage{mathpartir} %
\usepackage{thm-restate}
\usepackage{extarrows} %
\usepackage{algorithm2e}
\usepackage{bm}
\usepackage[inline]{enumitem}
\usepackage{thmtools, thm-restate}

\colorlet{instcolor}{Blue}

\usepackage{listings}
\lstdefinelanguage[RISC-V]{Assembler}
{
  alsoletter={.}, %
  alsodigit={0x}, %
  morekeywords=[1]{ %
    lb, lh, lw, lbu, lhu,
    sb, sh, sw,
    sll, slli, srl, srli, sra, srai,
    add, addi, sub, lui, auipc,
    xor, xori, or, ori, and, andi,
    slt, slti, sltu, sltiu,
    beq, bne, blt, bge, bltu, bgeu,
    j, jr, jal, jalr, ret,
    scall, break, nop, seqz, not,
  },
  morekeywords=[2]{ %
    .align, .ascii, .asciiz, .byte, .data, .double, .extern,
    .float, .globl, .half, .kdata, .ktext, .set, .space, .text, .word
  },
  morekeywords=[3]{ %
    zero, ra, sp, gp, tp, s0, fp,
    x0, x1, x2,
    t0, t1, t2, t3, t4, t5, t6,
    s1, s2, s3, s4, s5, s6, s7, s8, s9, s10, s11,
    a0, a1, a2, a3, a4, a5, a6, a7,
    ft0, ft1, ft2, ft3, ft4, ft5, ft6, ft7,
    fs0, fs1, fs2, fs3, fs4, fs5, fs6, fs7, fs8, fs9, fs10, fs11,
    fa0, fa1, fa2, fa3, fa4, fa5, fa6, fa7
  },
  morecomment=[l]{\#},  %
  morestring=[b]"      %
}

\lstdefinelanguage[myasm]{RISC-V}[RISC-V]{Assembler}
{
  alsoletter={.}, %
  alsodigit={0x}, %
  morekeywords=[1]{ %
    br, jmp, j, if, else, then, load, store
  },
  morekeywords=[3]{ %
    c, d, e, secret, a, x, y, r
  },
  keywords=[4]{
    s.br, s.call,
    lo.br, lo.call, lo.ret,
    lo.j
  },
}

\definecolor{mauve}{rgb}{0.58,0,0.82}

\lstdefinestyle{myasmstyle}{
  language={[myasm]RISC-V},
  keywordstyle=[4]\bfseries\color{instcolor},        %
  numbers=none,
  numberstyle=\tiny,
  frame=lines,                                      %
  xleftmargin=0em,
  framexleftmargin=0em,
  escapeinside={<@}{@>},
  mathescape
}

\usepackage{tikz}
\usetikzlibrary{positioning, quotes}

\usepackage{newfloat}
\DeclareFloatingEnvironment[
fileext = los,
placement={!ht},
name = Listing
]{listing}

\crefname{sublisting}{Listing}{Listings}
\Crefname{sublisting}{Listing}{Listings}
\DeclareCaptionSubType[alph]{listing}
\makeatletter
\renewcommand\p@subfigure{}
\makeatother

\newif\iftodos
\todosfalse

\newif\iftechreport
\techreporttrue

\iftodos
\newcommand\hans[1]{\textcolor{blue}{Hans: #1}}
\newcommand\marton[1]{\textcolor{magenta}{Marton: #1}}
\newcommand\lesly[1]{\textcolor{orange}{Lesly: #1}}
\newcommand\frank[1]{\textcolor{purple}{Frank: #1}}
\else
\newcommand\hans[1]{}
\newcommand\marton[1]{}
\newcommand\lesly[1]{}
\newcommand\frank[1]{}
\fi

\newcommand{\ie}{i.e.,\xspace}
\newcommand{\eg}{e.g.,\xspace}
\newcommand{\cf}{cf.\xspace}
\newcommand{\wrt}{w.r.t.\xspace}
\newcommand{\aka}{a.k.a.\xspace}

\newcommand\angles[1]{\langle #1 \rangle}

\newcommand\func[1]{\textsl{#1}}
\newcommand\listinglabel[1]{\lstinline[style=myasmstyle]!#1!}
\newcommand\bb[1]{\ensuremath{B_{\text{\listinglabel{#1}}}}}
\newcommand\bbidx[1]{\func{idx}(#1)}
\newcommand\entry[1]{\func{entry}(#1)}
\newcommand\exit[1]{\func{exit}(#1)}
\newcommand\proofcase[1]{\medskip\noindent\textit{#1}}

\crefname{algocf}{Algorithm}{Algorithms}
\Crefname{algocf}{Algorithm}{Algorithms}
\crefname{hypothesis}{Hyp.}{Hyp.}
\Crefname{hypothesis}{Hypothesis}{Hypotheses}

\usepackage{pifont}%
\newcommand{\cmark}{\text{\ding{51}}}%
\newcommand{\xmark}{\text{\ding{55}}}%
\newcommand{\ccmark}{\textcolor{Green}{\cmark}}%
\newcommand{\cxmark}{\textcolor{Red}{\xmark}}%

\newcommand{\mi}[1]{\ensuremath{\mathit{#1}}}
\newcommand{\mr}[1]{\ensuremath{\mathrm{#1}}}

\newcommand{\mf}[1]{\ensuremath{\mathbf{#1}}}

\newcommand{\ms}[1]{\ensuremath{\mathsf{#1}}}

\newcommand{\neutcol}[0]{black}
\newcommand{\stlccol}[0]{RoyalBlue}
\newcommand{\ulccol}[0]{RedOrange}

\newcommand{\col}[2]{\ensuremath{{\color{#1}{#2}}}}
\newcommand{\neut}[1]{{\mi{\col{\neutcol}{\protect\mbox{\unboldmath{$#1$}}}}}}

\newcommand{\src}[1]{\ms{\col{\stlccol}{#1}}}
\newcommand{\srcu}[1]{{\color{\stlccol}\ensuremath{\mbox{\unboldmath{\(\mathsf{#1}\)}}}}}
\newcommand{\trgb}[1]{\mf{\bm{\mr{\col{\ulccol }{#1}}}}}
\newcommand{\trg}[1]{{\mf{\col{\ulccol }{#1}}}}

\newcommand\requirements[2]{\hyperlink{#1}{\textcolor{black}{\textbf{#2}}}}
\newcommand\srbranch{\requirements{sreq:branches}{SR1}}
\newcommand\srbalanced{\requirements{sreq:balanced}{SR2}}
\newcommand\srfold{\requirements{sreq:fold}{SR3}}

\newcommand\hrcontract{\requirements{hreq:contract}{HR1}}
\newcommand\hrslice{\requirements{hreq:slice_granularity}{HR2}}
\newcommand\hruniform{\requirements{hreq:uniform_access_pattern}{HR2a}}
\newcommand\hrsafe{\requirements{hreq:safe}{HR3}}

\SetKwProg{Fn}{def}{\string:}{}
\SetKwFunction{FFoldLevel}{\textnormal{\func{FoldLevel}}}
\SetKwFunction{FFoldRegion}{\textnormal{\func{FoldRegion}}}
\SetKwFunction{FFoldFunct}{\textnormal{\func{FoldFunction}}}
\SetKwFunction{FInsertLob}{\textnormal{\func{RewriteTerminator}}}
\SetKwFunction{FInsertLocall}{\textnormal{\func{RewriteCall}}}
\SetKw{Assert}{assert}
\SetKwFor{For}{for}{:}{}
\SetKwSwitch{Switch}{Case}{Other}{switch}{:}{case}{otherwise}{}{}
\SetKwIF{If}{ElseIf}{Else}{if}{:}{else if}{else}{:}%
\SetKwInput{KwData}{Input}
\SetKwComment{Comment}{/* }{ */}

\newcommand\ocinterval[1]{[#1)}
\newcommand\mydef{\triangleq}

\newcommand\inst[2]{\def\param{#2}\textcolor{instcolor}{\texttt{\textnormal{#1}}}\ifx\param\empty{}\else{\textnormal{ #2}}\fi}
\newcommand\register[1]{\texttt{#1}}
\newcommand\qinst[3]{\def\param{#3}\textcolor{instcolor}{\texttt{#1}\textbf{\texttt{.#2}}}\ifx\param\empty{}\else{~\neut{\textnormal{#3}}}\fi}

\newcommand\lvl{L}
\newcommand\libra{Libra}

\newcommand\lob{\qinst{\qlo}{br}{}}
\newcommand\tlob{\qinst{\qtlo}{br}{}}
\newcommand\sasm[0]{\ensuremath{\srcu{asm}}}
\newcommand\tasm[0]{\ensuremath{\trgb{asm}}}
\newcommand\transform[0]{\ensuremath{\mathcal{F}}}

\newcommand\progmacro[2]{\def\param{#2}\ifx\param\empty{#1}\else{{#1[#2]}}\fi}
\newcommand\prog[1]{\progmacro{P}{#1}}
\newcommand\sprog[1]{\progmacro{\srcu{P}}{#1}}
\newcommand\tprog[1]{\progmacro{\trgb{P}}{#1}}

\newcommand\Regs{\mathbb{R}}
\newcommand\Vals{\mathbb{V}}
\newcommand\val{\ensuremath{v}}
\newcommand\uop{op\textsubscript{1}}
\newcommand\bop{op\textsubscript{2}}

\newcommand\Loc{\ensuremath \mathbb{L}}
\newcommand\loc{\ensuremath \ell}
\newcommand\sloc{\srcu{\loc}}
\newcommand\tloc{\trgb{\loc}}

\newcommand\qualfont[1]{\textcolor{instcolor}{\textbf{\texttt{#1}}}}
\newcommand\qlo{\qualfont{lo}}
\newcommand\qtlo{\qualfont{tlo}}
\newcommand\qs{\qualfont{s}}

\newcommand\qq{\neut{\mi{q}}} %

\newcommand\aconf[1]{\langle #1 \rangle}
\newcommand\mem{m}
\newcommand\reg{r}

\newcommand\pc{\texttt{pc}}
\newcommand\spc{\srcu{\pc}}
\newcommand\tpc{\trgb{\pc}}
\newcommand\tmem{\trgb{\mem}}
\newcommand\treg{\trgb{\reg}}
\newcommand\smem{\srcu{\mem}}
\newcommand\sreg{\srcu{\reg}}

\newcommand\retstack{\rho}
\newcommand\sretstack{\srcu{\retstack}}
\newcommand\tretstack{\trgb{\retstack}}
\newcommand\conf[1]{\ensuremath\def\param{#1}\ifx\param\empty{\sigma}\else{{\langle #1 \rangle}}\fi}

\newcommand\sconf[0]{\srcu{\sigma}}
\newcommand\tconf[0]{\trgb{\sigma}}

\newcommand\ctx{\texttt{ctx}}

\newcommand\ctxstack{\lambda}
\newcommand\tctxstack{\trgb{\ctxstack}}
\newcommand\bbc{\texttt{bbc}}
\newcommand\off{\texttt{off}}
\newcommand\curOff{\texttt{off}}
\newcommand\offt{\texttt{off}_{\texttt{t}}}
\newcommand\offf{\texttt{off}_{\texttt{f}}}

\newcommand\obs{\ensuremath{o}}
\newcommand\sobs{\ensuremath{\srcu{o}}}
\newcommand\tobs{\ensuremath{\trgb{o}}}

\newcommand\obsfunc{\textsl{obs}}
\newcommand\obsfuncweak{\textsl{obs}^{-}}
\newcommand\obsfuncstrong{\textsl{obs}^{+}}
\newcommand\id[1]{\textcolor{YellowOrange}{\textnormal{\textsl{#1}}}}
\newcommand\indist[0]{\ensuremath{\simeq}}

\newcommand\oni[1]{#1\textnormal{-ONI}}
\newcommand\class[1]{\mathscr{L}\def\param{#1}\ifx\param\empty{}\else(#1)\fi}
\newcommand\Instr[0]{\mathbb{I}}
\newcommand\IuSafe[0]{\Instr^{\ccmark}}
\newcommand\IbSafe[0]{\Instr^{\ccmark\text{-}\ccmark}}
\newcommand\IuUnsafe[0]{\Instr^{\cxmark}}
\newcommand\IlUnsafe[0]{\Instr^{\cxmark\text{-}\ccmark}}
\newcommand\IrUnsafe[0]{\Instr^{\ccmark\text{-}\cxmark}}
\newcommand\IlrUnsafe[0]{\Instr^{\cxmark\text{-}\cxmark}}

\newcommand\CorrRelation[2]{\overset{\text{\tiny{#1}}}{\sim}_{#2}}
\newcommand\libraeq[1]{\CorrRelation{#1}{}}
\newcommand\libraeqstrong[1]{\CorrRelation{#1}{+}}
\newcommand\corrloc[1]{\CorrRelation{#1}{\loc}}
\newcommand\CorrectStack[1]{\CorrRelation{#1}{\text{\tiny\ctx}}}

\newcommand\eval[1]{\xLongrightarrow{\neut{#1}}}
\newcommand\evalconf[4]{#1 {\hspace{3pt}}{\eval{#3}\hspace{-1pt}^{#4}} #2}
\newcommand\isaeval[3]{#2 {\hspace{3pt}}{\xrightarrow[]{#1} {\!}} #3}
\newcommand\seval[1]{\srcu{\eval{#1}}}
\newcommand\sevalconf[4]{#1 {\hspace{3pt}}{\seval{#3}\hspace{-1pt}^{#4}} #2}
\newcommand\teval[1]{\trgb{\eval{#1}}}
\newcommand\tevalconf[4]{#1 {\hspace{3pt}}{\teval{#3}\hspace{-1pt}^{#4}} #2}
\newcommand\eeval[2]{\llparenthesis #1 \rrparenthesis_{#2}}

\newtheorem{lemma}{Lemma}
\theoremstyle{definition}

\newtheorem{hypothesis}{Hypothesis}
\newtheorem{definition}{Definition}

\newtheorem{example}{Example}

\usetikzlibrary{positioning,decorations.pathreplacing,fit}

\newcommand{\tikzmark}[2][]{%
  \tikz[remember picture,overlay,baseline=-.5ex] \node[#1] (#2) {};%
}

\newcommand\pcupdatenct{
  \infer[pc-update]%
  {%
    \prog{\pc} = inst\\
    inst \not\in \{
\text{\qinst{\qq}{br}{}}, \text{\qinst{\qq}{call}{}},
\text{\inst{ret}{}}
\} \\
    \isaeval{inst}{\aconf{\mem, \reg}}{\aconf{\mem', \reg'}}\\
    \pc' = \pc + \bbc\\
  }{%
    \tevalconf
    {\conf{\mem, \reg, \pc, \retstack, \ctxstack \cdot (\bbc, \curOff)}}%
    {\conf{\mem', \reg', \pc', \retstack, \ctxstack \cdot (\bbc, \curOff)}}%
    {\obs{}}{}%
  }%
}

\newcommand\lobranchtrue{
  \infer[lob-true]%
  {%
    \prog{\pc} = \qinst{\qlo}{br}{{\(e\ \offt\ \offf\ \bbc'\)}}\\
    \eeval{\text{e}}{\reg} \neq 0\\
    \pc' = \func{next\_slice}(\pc, \bbc, \curOff) + \offt
  }{%
    \tevalconf
    {\conf{\mem, \reg, \pc, \retstack, \ctxstack \cdot (\bbc, \curOff)}}%
    {\conf{\mem, \reg, \pc', \retstack, \ctxstack \cdot (\bbc', \offt)}}%
    {\obs}{}%
  }%
}

\newcommand\locall{
  \infer[lo-call]%
  {%
    \prog{\pc} = \qinst{\qlo}{call}{{\(b\ \ell\)}}\\
    \off' = (\textsl{if } b = \top \textsl{ then } 0 \textsl{ else } 1)\\
    \pc' = \ell + \off'\\
    \retstack' = \retstack \cdot \pc + \bbc\\
    \ctxstack' = \ctxstack \cdot (\bbc, \curOff) \cdot (2,\off')
  }{%
    \tevalconf%
    {\conf{\mem, \reg, \pc, \retstack, \ctxstack \cdot (\bbc, \curOff)}}%
    {\conf{\mem, \reg, \pc', \retstack', \ctxstack'}}%
    {\obs{}}{}%
  }%
}

\newcommand\return{
  \infer[ret]%
  {%
    \prog{\pc} = \inst{ret}{}\\
    \pc' = \loc\\
  }{%
    \tevalconf%
    {\conf{\mem, \reg, \pc, \retstack \cdot \loc, \ctxstack \cdot (\bbc, \curOff)}}%
    {\conf{\mem, \reg, \pc', \retstack, \ctxstack}}%
    {\obs{}}{}%
  }%
}

\AtBeginDocument{%
  }

\acmConference[CCS '24]{Proceedings of the 2024 ACM SIGSAC Conference on Computer and Communications Security}{October 14--18, 2024}{Salt Lake City, UT, USA}
\acmBooktitle{Proceedings of the 2024 ACM SIGSAC Conference on Computer and Communications Security (CCS '24), October 14--18, 2024, Salt Lake City, UT, USA}

\iftechreport
\setcopyright{none}
\renewcommand\footnotetextcopyrightpermission[1]{}
\makeatletter
\renewcommand\@formatdoi[1]{\ignorespaces}
\makeatother
\else
\copyrightyear{2024}
\acmYear{2024}
\setcopyright{acmlicensed}
\acmDOI{10.1145/3658644.3690319}
\acmISBN{979-8-4007-0636-3/24/10}
\fi

\begin{document}

\iftechreport
\title[Libra: Principled, Secure and Efficient Balanced Execution]{Libra: Architectural Support For Principled, Secure And Efficient Balanced Execution On High-End Processors (Extended Version)}
\fi

\author{Hans Winderix}
\orcid{0000-0002-0165-7915}
\email{hans.winderix@kuleuven.be}
\affiliation{%
 \institution{DistriNet, KU Leuven}
 \city{Leuven}
 \country{Belgium}
 \postcode{3001}
}

\author{Marton Bognar}
\orcid{0000-0002-8641-7549}
\email{marton.bognar@kuleuven.be}
\affiliation{%
 \institution{DistriNet, KU Leuven}
 \city{Leuven}
 \country{Belgium}
 \postcode{3001}
}

\author{Lesly-Ann Daniel}
\orcid{0000-0002-2772-3722}
\email{lesly-ann.daniel@kuleuven.be}
\affiliation{%
 \institution{DistriNet, KU Leuven}
 \city{Leuven}
 \country{Belgium}
 \postcode{3001}
}

\author{Frank Piessens}
\orcid{0000-0001-5438-153X}
\email{frank.piessens@kuleuven.be}
\affiliation{%
 \institution{DistriNet, KU Leuven}
 \city{Leuven}
 \country{Belgium}
 \postcode{3001}
}

\begin{abstract}

  \Ac{CFL} attacks enable an attacker to expose control-flow decisions of a victim program via side-channel observations.
  \emph{Linearization} (\ie{} elimination) of secret-dependent control flow is the main countermeasure against these attacks, yet it comes at a non-negligible cost.
  Conversely, \emph{balancing} secret-dependent branches often incurs a smaller overhead, but is notoriously insecure on high-end processors.
  Hence, linearization has been widely believed to be \emph{the only} effective countermeasure against \ac{CFL} attacks.
  In this paper, we challenge this belief and investigate an unexplored alternative: how to securely balance secret-dependent branches on higher-end processors?

  We propose \libra{}, a generic and principled hardware-software codesign to efficiently address \ac{CFL} on high-end processors.
  We perform a systematic classification of hardware primitives leaking control flow from the literature, and provide guidelines
  to handle them with our design. Importantly, \libra{} enables secure control-flow balancing without the need to disable performance-critical hardware such as the instruction cache and the prefetcher.
  We formalize the semantics of \libra{} and propose a code transformation algorithm for securing programs, which we prove correct and secure.
  Finally, we implement and evaluate \libra{} on an out-of-order RISC-V processor, showing performance overhead on par with insecure balanced code, and outperforming state-of-the-art linearized code by 19.3\%.
\end{abstract}

\maketitle

\section{Introduction}

In recent years, software-based microarchitectural attacks~\cite{ge:2018-survey, lou:2021-a} have emerged as a critical security threat.
When multiple stakeholders run code on the same computing device, this type of side-channel attack makes it possible for an attacker to infer program secrets just by monitoring from software how a victim uses shared hardware such as the cache, branch predictor, or prefetcher.

Of special interest to this work are so-called \emph{\acf{CFL} attacks}~\cite{bognar:2022-mind,chen:2023-afterimage,lee:2017-inferring,molnar:2005-program,puddu:2021-frontal,vanbulck:2018-nemesis,yu:2023-all} whereby an attacker tries to expose the \ac{PC} trace of a victim program via side-channel observations with the aim of revealing the outcome of conditional control-flow decisions.
The program's conditional control flow exposes the outcome of the condition that determines the control flow, which poses a security threat if that condition depends on secret information.
In the presence of a microarchitectural attacker, a program's control flow can, in general, be observed in the microarchitectural state of shared hardware or through contention.

A possible software countermeasure against \ac{CFL} attacks is \emph{control-flow balancing}~\cite{agat:2000-transforming,bognar:2023-microprofiler,dewald:2017-avr,kopf:2007-transformational,pouyanrad:2020-scf,winderix:2021-compiler}, a program transformation which aims to make the execution of all possible targets of a control-transfer instruction appear the same to an attacker.
So far, control-flow balancing has been shown to be secure only for a class of low-power embedded processors~\cite{bognar:2023-microprofiler,winderix:2021-compiler}.
This is because modern superscalar processors feature critical performance-enhancing hardware that maintains state as a function of the \ac{PC}, thus leaking the \ac{PC} in an unbalanceable way when this hardware is shared between different security domains.
For this reason, it is widely accepted that, to counter \ac{CFL} attacks on higher-end processors, programs must be \ac{PC}-secure~\cite{molnar:2005-program}, i.e., their \ac{PC}  should be independent from secret information. \ac{PC}-secure programs are created by avoiding secret-dependent control flow and the techniques for doing so are well-documented in the literature~\cite{borello:2021-constantine,molnar:2005-program,rane:2015-raccoon,vanoverloop:2024-compiler,wu:2018-eliminating}.

Unfortunately, this advice has not been questioned much.
Over the years, it has been evolving into a dogma and it has become an established practice to hardcode it in \emph{constant-time}~\cite{almeida:2016-verifying} source code, preventing the adoption of more relaxed policies (for simpler architectures or for weaker attacker models). Furthermore, this trend creates the fallacy that secret-dependent control flow is inherently insecure and, consequently, it discourages the search for novel mechanisms to securely execute \ac{PC}-insecure programs on higher-end processors.

On the other hand, there still exists a strong desire to keep the secret-dependent control flow for performance reasons, even on high-end processors.
Vendors of cryptographic libraries, for instance, sometimes take the risk and do balance secret-dependent branches~\cite{yu:2023-all} instead of eliminating them.
As another example, numerous offensive research papers have been published that develop new \ac{CFL} attacks, accompanied by ad-hoc defenses, which are later found to be vulnerable by other offensive research, a trend that has been recently described as \emph{the \ac{CFL} arms race}~\cite{yu:2023-all}.

\paragraph{Our Proposal}
In this work, we challenge the widely-held belief that secret-dependent control flow is inherently insecure on high-end processors and propose a well-founded hardware-software codesign for secure and efficient balanced execution.
In contrast to prior works that target a single vulnerability and propose ad-hoc, incremental defenses, we propose a principled solution that addresses the \ac{CFL} problem in a generic way with the goal of ending the \ac{CFL} arms race. Also in contrast to prior works, we do not assume a simple processor pipeline and scheduling but support modern out-of-order processor designs.

We conduct a rigorous analysis of how hardware optimizations leak a program's control flow.
A key finding is that hardware optimizations can be partitioned into two categories; those that yield \emph{balanceable observations} and those that yield \emph{unbalanceable observations}.
Balanceable observations can be securely balanced by software-only approaches.
Unbalanceable observations require hardware support.
Based on the findings of our analysis, we propose \libra{}, a hardware-software security contract that lays the principled foundation for secure balanced execution.
We introduce a novel memory layout, called \emph{folded layout}, and an algorithm for \emph{folding} balanced code regions, which makes it possible to keep enabled performance-critical hardware optimizations without compromising security.
Additionally, we propose an ISA extension for executing folded regions.

In a nutshell, we make the following contributions:
\begin{itemize}
  \item
    A novel hardware-software contract, called \libra{}, for secure and efficient balanced execution (Section~\ref{sec:libra-overview}).
  \item
    A formalization of the ISA-level semantics of \libra{} and security and correctness proofs of our folding algorithm (Section~\ref{sec:libra-formal-semantics}).
  \item
    A characterization of hardware optimizations regarding how they leak a program's control flow (Section~\ref{sec:characterization}).
  \item
    Recommendations for hardware designers wishing to adopt \libra{} to their designs (Section~\ref{sec:characterization}).
  \item
    An implementation of \libra{} on an out-of-order RISC-V core (Section~\ref{sec:implementation}).
  \item
    An experimental evaluation showing that balanced execution is secure and efficient at a low hardware cost (Section~\ref{sec:evaluation}).
\end{itemize}

\paragraph{Additional material}
Our RISC-V implementation and evaluation are archived on Zenodo~\cite{libra-zenodo} and available on GitHub: \url{https://github.com/proteus-core/libra}.

\section{Terminology and Background}

\subsection{Terminology}
We first define relevant terminology from the fields of graph theory and compiler construction and then introduce some new vocabulary (marked with \(*\)).

\begin{definition}[Basic block]\label{def:basic-block}
A basic block is a straight-line instruction sequence always entered at the beginning and exited at the end.
\end{definition}

\begin{definition}[Control-flow graph]\label{def:control-flow-graph}
  A \ac{CFG} is a directed graph that represents all the paths that might be traversed through a program during its execution. The nodes of a \ac{CFG} represent basic blocks, the edges represent control-flow transfers.
\end{definition}

Without loss of generality, we assume that a \ac{CFG} has a unique \emph{entry} and a unique \emph{exit} block.
We also assume that the last instruction in a basic block is a control-transfer instruction, which designates the possible successor blocks.
We refer to this instruction as the \emph{terminating instruction} of the basic block.
\Cref{fig:level-structure} contains an illustration of a \ac{CFG} with \(\bb{En}\) the entry basic block and \(\bb{Ex}\) the exit basic block.

\begin{figure}
  \begin{minipage}{.45\linewidth}
\begin{lstlisting}[
  style=myasmstyle,
  ]
En: br a0,t,f         <@\tikzmark{bgn1}\tikzmark{end1}@>
 t:    br a1,tt,tf    <@\tikzmark{bgn2}\tikzmark{end2}@>
tt:     add s1,s2,s3  <@\tikzmark{bgn3}@>
        j Ex          <@\tikzmark{end3}@>
tf:     add s2,s3,s4  <@\tikzmark{bgn4}@>
        j Ex          <@\tikzmark{end4}@>
 f:    sub s1,s2,s3   <@\tikzmark{bgn5}@>
       j Ex           <@\tikzmark{end5}@>
Ex: [...]             <@\tikzmark{bgn6}\tikzmark{end6}@>
\end{lstlisting}
  \end{minipage}
  \hfill
  \begin{minipage}{.45\linewidth}
    \begin{tikzpicture}[
    node distance = 1.0em and .2em,
       box/.style = {rectangle, draw, fill=#1,
                     minimum width=7mm, minimum height=5mm}
                        ]
\node (A) [box=white] {\(\bb{En}\)};
\node (B) [box=white,below left=of A] {\(\bb{t}\)};
\node (C) [box=white,below right=of A] {\(\bb{f}\)};
\node (D) [box=white,below left=1.0em and -.8em of B] {\(\bb{tt}\)};
\node (E) [box=white,below right=1.0em and -.8em of B] {\(\bb{tf}\)};
\node (F) [box=white,below =of E] {\(\bb{Ex}\)};

\draw[->] (A) to (B);
\draw[->] (A) to (C);
\draw[->] (B) to (D);
\draw[->] (B) to (E);
\path [->] (C) edge[bend left=30] node {} (F);
\draw[->] (D) to (F);
\draw[->] (E) to (F);
\end{tikzpicture}
  \end{minipage}
  \caption{A program and its \ac{CFG}.}\label{fig:level-structure}
\end{figure}

\begin{definition}[Distance]\label{def:distance}
  The distance between two basic blocks in a \ac{CFG} is the number of edges in a shortest path connecting them.
\end{definition}

In \cref{fig:level-structure}, the distance between the basic blocks \(\bb{En}\) and \(\bb{Ex}\) is 2 (\(\bb{En}\rightarrow\bb{f}\rightarrow\bb{Ex}\)).
The distance between two instructions is defined similarly by considering individual instructions as basic blocks.

\begin{definition}[Postdominance]\label{def:postdominance}
 A basic block \(Y\) postdominates a basic block \(X\) (\ie{} \(Y\) is a postdominator of \(X\)) if all paths from \(X\) to the exit block go through \(Y\).
\end{definition}
The closest postdominator of a basic block is called its \emph{immediate postdominator}.
In \cref{fig:level-structure}, basic block \(\bb{Ex}\) postdominates basic block \(\bb{En}\). It is also the immediate postdominator of \(\bb{En}\).

\begin{definition}[Level structure]\label{def:level-structure}
  The level structure of a \ac{CFG} is a partition of the basic blocks into subsets (levels) that have the same distance from the entry basic block.
\end{definition}

The level structure of the \ac{CFG} in \cref{fig:level-structure} consists of three levels: $L_0 = \left\{\bb{En} \right\}, L_1 = \left\{\bb{t},\bb{f}\right\}, L_2 = \left\{\bb{tt}, \bb{tf}, \bb{Ex}\right\}$.

\begin{definition}[\(*\)Level slice]\label{def:level-slice}
  The set of equidistant instructions for a distance $\delta$ with respect to basic block \(B\) forms the level slice (or simply \emph{slice}) determined by the tuple (\(B\), $\delta$)
\end{definition}

In \cref{fig:level-structure}, the slice of distance 0 is
\(\{\text{\lstinline[style=myasmstyle]{br a0,t,f}}\}\) and the slice of distance 1 is
\(\{\text{\lstinline[style=myasmstyle]{br a1,tt,tf}};\  \text{\lstinline[style=myasmstyle]{sub s1,s2,s3}}\}\) (both relative to \(\bb{En}\)).

\begin{definition}[\(*\)Secret-dependent region]\label{def:secret-dependent-region1}
  The set of basic blocks between a secret-dependent control-transfer instruction $inst$ and its immediate postdominator form the secret-dependent region determined by $inst$.
\end{definition}
We refer to the basic block containing the secret-dependent control-transfer instruction as the \emph{entry block} of the region, and to its immediate postdominator as the \emph{exit block} of the region.
In \cref{fig:level-structure}, if \lstinline{a1} is secret (line \listinglabel{t}), then $\left\{\bb{tt}, \bb{tf}\right\}$ is the secret-dependent region determined by the instruction on line \listinglabel{t}.
The entry block of the region is \(\bb{t}\), the exit block \(\bb{Ex}\).
Similar to the level structure of a \ac{CFG}, we define the \emph{level structure of a secret-dependent region} as the partition of its basic blocks into subsets (levels) that have the same distance from the region's entry block.

\subsection{The Control-Flow Leakage Problem}

\subsubsection{Control-Flow Leakage Attacks}

\Ac{CFL} attacks are a type of microarchitectural attack whereby an attacker tries to learn the outcome of a secret-dependent branch by exposing the control flow via microarchitectural side channels.
Consider the program in \cref{running-vulnerable}.
When the branch on line~1 evaluates to true, the instructions on lines~2-3 are executed and the program exits.
When the branch evaluates to false, the instruction on line~4 is executed and the program exits.
An attacker that is able to observe the program's execution time will be able to distinguish the two executions, and hence learn if \lstinline[style=myasmstyle]{secret} evaluates to true or false.

\begin{listing}[ht]
  \centering
\addtocounter{lstlisting}{1} %
\captionof{listing}{Code vulnerable to CFL attacks (\cref{running-vulnerable}) and its balanced version (\cref{running-balanced}).}\label{lst:transient_execution_examples}
\vspace{-1em}
\begin{minipage}{.45\linewidth}
\begin{lstlisting}[
  style=myasmstyle,
  label={running-vulnerable},
  numbers={left},
  ]
     br secret,t,f
 t:   add s1,s2,s3
      j Ex
 f:   add s2,s3,s4

Ex: [...]
\end{lstlisting}
\subcaption[listing]{}\label{running-vulnerable}
\end{minipage}
\hfill
\begin{minipage}{.45\linewidth}
\begin{lstlisting}[
  style=myasmstyle,
  label=running-balanced,
  numbers={left},
  ]
     br secret,t,f
 t:   add s1,s2,s3
      j Ex
 f:   add s2,s3,s4
      j Ex
Ex: [...]
\end{lstlisting}
\subcaption[listing]{}\label{running-balanced}
\end{minipage}
\vspace{-1.6em}
\end{listing}

Besides this start-to-end timing difference, interrupt latency~\cite{vanbulck:2018-nemesis}, data cache contention~\cite{osvik:2006-cache}, structural dependencies~\cite{aldaya:2019-port} or data dependencies stalling the pipeline are other examples of microarchitectural events that can be monitored by an attacker to leak the control flow.
Consider \cref{running-vulnerable} again and assume that the addresses of the \lstinline[style=myasmstyle]{add} instructions (lines 2 and 4) map to different instruction cache lines.
Monitoring which cache line has been touched (for instance with the Flush+Reload attack~\cite{yarom:2014-flush}) will reveal the control flow.

Two common software countermeasures against CFL attacks are \emph{control-flow balancing} and \emph{control-flow linearization}.
The former technique keeps the secret-dependent control flow intact while the latter eliminates it completely.

\subsubsection{Control-Flow Balancing}\label{sec:control-flow-balancing}
Control-flow balancing is based on the idea that if the two sides of a secret-dependent branch induce exactly the same attacker-observable behavior, then executing the code does not reveal via side channels which side of the branch has been executed.
\Cref{running-balanced} gives the balanced form of \cref{running-vulnerable}. The \lstinline[style=myasmstyle]{add} instruction on line~2 is balanced with the \lstinline[style=myasmstyle]{add} instruction on line~4 and a jump instruction is added to the \listinglabel{f} path on line~5 to balance it with the jump on line~3 in the \listinglabel{t} path.

Recent work~\cite{bognar:2023-microprofiler, winderix:2021-compiler} has demonstrated the security (and efficiency) of control-flow balancing for small, embedded processors with deterministic timing behavior.
The authors propose a methodology consisting of three steps.
First, by profiling the microarchitecture, the instruction set is classified into a number of \emph{leakage classes} such that executing instructions from the same leakage class induces the same side-channel observations.
Second, a dummy (no-op) instruction is composed for every leakage class.
Lastly, the secret-dependent branches are algorithmically balanced~\cite{winderix:2021-compiler} with respect to the leakage classification, and by inserting dummy instructions when necessary.
This approach ensures that the dynamic instruction trace of balanced code always produces the same sequence of leakage classes.

Although control-flow balancing counters attacks exploiting microarchitectural optimizations on low-end devices~\cite{moghimi:2020-copycat, vanbulck:2018-nemesis}, higher-end devices (the target of our work) typically feature optimizations yielding observations that are unbalanceable in software alone.
Yet, for performance reasons, balanced control flow is sometimes found in security-critical libraries targeting these devices~\cite{moghimi:2020-copycat,yu:2023-all}.
Thus, how to make balanced execution secure on these higher-end devices remains an important research question.

\subsubsection{Control-Flow Linearization}\label{sec:control-flow-linearization}

{}

Control-flow linearization is a key principle of the widely-established constant-time programming discipline~\cite{almeida:2016-verifying}.
By eliminating secret-dependent branches, control-flow linearization ensures that the \ac{PC} does not get tainted (i.e., that the \ac{PC} trace is independent of secrets).
Several linearization techniques have been proposed in the literature~\cite{borello:2021-constantine,molnar:2005-program,rane:2015-raccoon,soares:2023-side,vanoverloop:2024-compiler,wu:2018-eliminating}.
\Cref{running-linearized} contains the linearized form of the running example from \cref{running-vulnerable}, based on a state-of-the-art method that was first proposed by Molnar et al.~\cite{molnar:2005-program}.
Compared with the balanced form from \cref{running-balanced}, the linearized form comes with a higher cost due to the use of additional instructions and registers.

{}

\begin{lstlisting}[
  caption={Linearized form of the vulnerable code in \cref{running-vulnerable}.},
  style=myasmstyle,
  label=running-linearized,
  numbers={left},
]
seqz t1,secret

addi t1,t1,-1  # t1 = true mask  (in {0xffff, 0x0000})
not  t2,t1     # t2 = false mask (in {0xffff, 0x0000})
and  t3,s1,t1  # start of else
add  s1,s2,s3
and  s1,s1,t2
or   s1,s1,t3  # start of then
and  t3,s2,t2
add  s2,s3,s4
and  s2,s3,t1
or   s2,s2,t3
\end{lstlisting}

\subsubsection{This paper}

The goal of this work is to make sure that executing balanced code (which contains secret-dependent control flow) on high-end processors does not leak more information than executing the equivalent linearized code (which does \emph{not} contain secret-dependent control flow).
We demonstrate that, with minimal hardware support, it is possible to securely balance secret-dependent control flow on higher-end platforms, %
without disabling performance-critical hardware resources that are shared between different stakeholders.

\section{Threat Model}\label{sec:assumptions}

We consider an adversary with the goal to infer secrets (e.g., cryptographic keys) by learning the secret-dependent control flow of a victim application.
We consider an adversary with the same capabilities as an adversary under the \emph{classic} constant-time threat model, and thus assume that applications are hardened against transient execution attacks~\cite{canella:2019-systematic}.
More specifically, an adversary with the capabilities of this threat model is able to run arbitrary code alongside an architecturally isolated victim (e.g., via process isolation) on the same machine and it shares hardware resources, such as the branch predictor, cache hierarchy and execution units with the victim.
This setting enables the adversary to precisely observe the execution time of the victim, and how it uses the shared resources.
If these observations depend on the secret control flow, the adversary is able to learn something about the secret.

We consider software-based timing channels, i.e., the adversary monitors the microarchitectural resource usage via timers from software~\cite{ge:2018-survey, lou:2021-a}.
Side channels that require physical access and physical equipment to measure quantities such as power consumption~\cite{kocher:1999-differential} or EM emissions~\cite{quisquater:2001-smart} are out of scope for this paper.
Similarly, other types of software-based side-channel attacks, such as software-based fault attacks~\cite{murdock:2020-plundervolt} and software-based power attacks~\cite{lipp:2021-platypus} are out of scope and subject of orthogonal mitigations.

We make no further assumptions on the type of (software-based) microarchitectural side-channels attacks that can be mounted by the adversary, ranging from classic cache attacks~\cite{osvik:2006-cache} to more recent contention-based attacks~\cite{aldaya:2019-port}.

\section{Overview of \libra}\label{sec:libra-overview}

A program's control flow can leak through observations induced by various microarchitectural optimizations.
Some of these observations, such as instruction latency, are independent of the value of the \ac{PC}.
We refer to optimizations yielding this type of observation as \emph{sources of balanceable leakage} as their observations can be balanced by software.
However, some performance-critical optimizations commonly found in modern hardware (e.g., the instruction cache and the instruction prefetcher) yield observations that are dependent on the value of the \ac{PC}. They inevitably leak the control flow.
We refer to these optimizations as \emph{sources of unbalanceable} leakage as they cannot be dealt with by software alone.
In \cref{sec:characterization}, we study this distinction further and provide a comprehensive characterization of hardware optimizations regarding how they leak the control flow.

Existing control-flow balancing solutions are ineffective against  unbalanceable leakage.
It is the goal of \libra{} to address this gap via a novel hardware-software security contract for secure and efficient balanced execution.
On the one hand, the software is responsible for balancing secret-dependent control flow under a \emph{weak observer mode} (accounting for the balanceable leakage) in which the \ac{PC} does not leak.
On the other hand, the hardware provides support to deal with the sources of unbalanceable leakage to ensure that the program remains secure in a \emph{strong observer mode}, representative of our threat model (\cref{{sec:assumptions}}) for high-end processors.

\subsection{Leakage Contract}\label{sec:libra-contract}

\libra{} requires the hardware to augment the ISA with a \emph{leakage contract}
that provides sufficient information on how to balance the control flow.
Software, such as a compiler, can then rely on this contract 1) to securely balance secret-dependent control flow (making control-flow balancing a \emph{principled} code transformation) or 2) to verify that secret-dependent control flow is securely balanced.
This stands in contrast to prior works~\cite{agat:2000-transforming,kopf:2007-transformational,bognar:2023-microprofiler,dewald:2017-avr,winderix:2021-compiler,pouyanrad:2020-scf}, where it is the responsibility of the software to empirically figure out \emph{how} to balance corresponding instructions.

The \libra{} leakage contract classifies an instruction set into two dimensions.
First, it partitions instructions into \emph{leakage classes}~\cite{bognar:2023-microprofiler,winderix:2021-compiler} such that instructions from the same leakage class yield identical side-channel observations.
Importantly, any instruction can be used to balance any other instruction from the same leakage class.
For every leakage class, the contract additionally designates a canonical \emph{dummy instruction}, which does not produce architectural effects (\eg{} \lstinline[style=myasmstyle]{mv x1, x1}).
Finally, the hardware provides a blocklist of instructions that are not supported in balanced regions. Blocklisted instructions have to be rewritten in terms of non-blocklisted instructions before performing control-flow balancing.

Second, the leakage contract partitions the instruction set into safe and unsafe instructions~\cite{yu:2019-data}.
\emph{Safe instructions} are instructions whose timing and shared microarchitectural resource usage are independent of the values of their operands.
For instance, an \lstinline[style=myasmstyle]{add} instruction is typically implemented in a safe way, while
a \lstinline[style=myasmstyle]{load} typically exposes the value of the address operand on systems with a data cache (making it an \emph{unsafe instruction}).
It is insecure to pass secrets to unsafe instructions
but it is secure to use unsafe instructions in balanced regions if it can be proven that the operands of any two equidistant unsafe instructions are the same for all possible executions.
For instance, the code \lstinline[style=myasmstyle]{if (secret)  load x0 a else load x1 a} is secure as the resulting observation is independent of \lstinline[style=myasmstyle]{secret} (under the assumption that the \lstinline[style=myasmstyle]{load} is only unsafe in its address operand).

\subsection{ISA Extension}\label{sec:hardware-extension}
The goal of \libra{} is to securely execute balanced code regions on high-end CPUs without disabling performance-critical optimizations. %
In particular, \libra{} aims at keeping \emph{all} modern hardware optimizations fully enabled when executing security-insensitive code (\ie the common case), and keeping \emph{as many optimizations as possible} in secret-dependent regions.

To this end, \libra{} proposes an ISA extension introducing two main novel features:
\begin{itemize}
  \item A novel memory layout for balanced code, termed \emph{folded layout}, which interleaves the instructions from balanced regions by placing the level slices sequentially in memory.
  \item A new instruction, the \emph{level-offset branch} (\lob), which informs the CPU how to navigate a folded region. Additionally, it signals to the CPU that it is about to execute a secret-dependent region such that it can adapt the behavior of some optimizations.
\end{itemize}

Importantly, even though folding
sequentially lays out instructions of balanced regions in memory (reminiscent of linearization), the \emph{original control flow of the program is preserved}, \ie only one side of a folded conditional branch is executed, as prescribed by the original \ac{CFG} (just like with standard code balancing).

The level-offset branch \qinst{\qlo}{br}{\(c,\offt:\offf:\bbc\)} specifies how to navigate a folded region:
\begin{enumerate}
  \item The level offsets \(\offt\) and \(\offf\) indicate what instructions of the next level to execute, depending on whether the condition \(c\) is true or false;
  \item The basic block count \(\bbc\) indicates the number of basic blocks of the next level (the slice size of the next level) and is used to increment the \ac{PC} by the correct value.
\end{enumerate}

\Cref{running-folded-1} illustrates how to fold the balanced code from \cref{running-balanced}. First, the two \inst{add}{} and the two \inst{j} instructions are sequentially placed in memory. Second, the conditional branch is rewritten using a \(\lob\) with \(\offt = 0\), \(\offf = 1\) and \(\bbc = 2\).
After the \(\lob\), the CPU will execute the folded region slice by slice, incrementing the \ac{PC} by \(2\).
If the condition is true, the first (offset \(\offt\)) instruction of each slice is executed, otherwise the second (offset \(\offf\)) instruction is executed.
Finally, the terminating \inst{j} instructions are replaced by \(\lob\) instructions to reset the level offset and \(\bbc\) and resume ``normal'' execution at the \listinglabel{Ex} label.

\noindent
\begin{minipage}{\linewidth}
\begin{lstlisting}[
  caption={Folded form of the balanced code in \cref{running-balanced}.},
  style=myasmstyle,
  label=running-folded-1
]
       lo.br secret,0:1:2   # offT:offF:bbc
   L1:   add s1,s2,s3
         add s2,s3,s4
         lo.br zero,0:0:1   # offT:offF:bbc
         lo.br zero,0:0:1   # offT:offF:bbc
   Ex: [...]
\end{lstlisting}
\end{minipage}

\paragraph{How does \libra{} address unbalanceable leakage?}
The design of Libra is tailored to address unbalanceable leakage in hardware efficiently, \ie by keeping essential hardware optimizations enabled.
Yet, to establish the security guarantees, \libra{} requires that the \ac{PC} does not leak at a finer granularity than a slice, possibly requiring adaptations to the behavior of some optimizations.

Importantly, the folded memory layout is crucial to keep enabled performance-critical optimizations of modern hardware (e.g., the instruction cache) without, or with only minimal, adaptations.
By virtue of folding (which creates a linear memory layout), the hardware can efficiently implement a data-oblivious instruction memory access pattern by always prefetching all the slices in the same order, effectively making it independent of the outcomes of conditional branch(es).

While some sources of unbalanceable leakage do not require hardware modifications, some will, possibly degrading performance.
However, because the hardware is informed when it is executing a folded region, these modifications can be limited to folded regions only.
For instance, some hardware structures, such as the branch predictor, must be disabled for the \lob{} instruction to prevent control-flow exposure to an attacker sharing the branch predictor.
However, the linear layout of a folded region makes the branch predictor unnecessary for \lob{} instructions, because there is no uncertainty (at slice granularity) what address the sequential prefetcher should fetch from, so it can fill the cache with the instructions that are about to be fetched by the CPU.

In \cref{sec:characterization}, we present, based on a rigorous study of the attack literature, a characterization of the sources of unbalanceable leakage (with folding in mind), and we provide guidelines about how to handle them.

\subsection{Advanced Features}\label{sec:advanced features}

\subsubsection{Nested branches}\label{sec:libra-secret-indepent-branches}
When folding a region with a nested branch (as in \cref{nested-branches}), the software must fold the level structure of the entire outer region,
as shown in \cref{nested-branches-folded}. %
The slice size grows with the level of nesting.
In the example from \cref{nested-branches-folded}, each slice of the second level consists of four instructions.
Recall that the hardware has to make sure to fetch instructions without exposing their offset within the current level. For instance, if a slice occupies multiple cache lines, the hardware must ensure to always touch all the cache lines in the same order, irrespective of the current instruction's offset.

\begin{listing}[ht]
  \centering
\setcounter{listing}{\value{lstlisting}} %
\addtocounter{lstlisting}{1} %
\captionof{listing}{Region with nested branches (\cref{nested-branches}) and its folded version (\cref{nested-branches-folded}).}
\vspace{-1em}
\begin{minipage}{.45\linewidth}
\begin{lstlisting}[
  style=myasmstyle,
  label=nested-branches
]
    br secret,t,f
 t:  br c,tt,tf
tt:    add r,r,4
       j Ex
tf:    add r,r,8
       j Ex
 f:  br c,ft,ff
ft:    sub r,r,4
       j Ex
ff:    sub r,r,8
       j Ex
Ex: [...]
\end{lstlisting}
\subcaption[listing]{}\label{nested-branches}
\end{minipage}
\hfill
\begin{minipage}{.45\linewidth}
\begin{lstlisting}[
  style=myasmstyle,
  label=nested-branches-folded
]
    lo.br secret,0:1:2
L1:   lo.br c,0:1:4
      lo.br c,2:3:4
L2:     add r,r,4
        add r,r,8
        sub r,r,4
        sub r,r,8
        lo.br zero,0:0:1
        lo.br zero,0:0:1
        lo.br zero,0:0:1
        lo.br zero,0:0:1
Ex: [...]
\end{lstlisting}
\subcaption[listing]{}\label{nested-branches-folded}
\end{minipage}
\label{lst:transform-br}
\vspace{-1.6em}
\end{listing}

Note that when a nested branch does not depend on secret information (\eg{} a loop with a constant trip count), it can be more efficient to keep the branch instead of folding it.
In that case, for correctness, the software must ensure that the level offsets of the target instructions are consistent regarding the offsets of the branch instructions.
Moreover, for security, the software must ensure that the branch targets of the branches in the source slice all point to targets in the same target slice.

\subsubsection{Function calls}\label{sec:libra-calls}
To support function calls in balanced code, prior work on control-flow balancing~\cite{bognar:2023-microprofiler, winderix:2021-compiler} proposed to create a dummy function for each function called from a secret-dependent region.
A dummy function is mostly made up of dummy (no-op) instructions designed to mirror the behavior of the real function.
These dummy instructions ensure that both the dummy and real functions cause identical changes in the microarchitectural state.
As a result, an attacker cannot distinguish between the execution of the dummy function and that of the real function.
A call to a function in a secret-dependent region can then be balanced with a call to its dummy version.
\libra{} supports this scheme, yet in order not to expose the control flow on higher-end CPUs (e.g., via the instruction cache), functions must be folded with their dummy counterpart.
\libra{} provides hardware support to efficiently invoke a folded function and extends the ISA with a new instruction, the \emph{level-offset call}: \qinst{\qlo}{call}{{\(b\ \ell\)}}.
The instruction jumps to the folded function and, according to the boolean immediate \(b\), either executes the real part or the dummy part of the folded function.
Additionally, the CPU must save/restore the \libra{} state (\ie{} current offset and \(\bbc\)) of the caller upon calls/returns.
\libra{} proposes a two-level hardware stack, used for storing and restoring the \libra{} state of the caller.
For non-leaf functions (i.e., to support more than one level of nesting, including recursion), the software is responsible to save and restore the \libra{} state on a software-based stack.

\subsubsection{Exceptions}

Instructions that may throw exceptions are inherently unsafe because whether an exception is thrown depends on the value of their operands and handling an exception impacts both the timing and resource usage of an application.
Therefore, such instructions should be treated similarly to other unsafe, balanceable instructions, by balancing the unsafe operands and their dependencies.

\subsection{Hardware-Software Security Contract}\label{sec:design-choices}

In summary, with \libra{} we propose a hardware-software security contract for balanced execution.
If both parties fulfill their part of the contract, then executing a balanced code region will not leak more information than the equivalent linearized region.

\subsubsection*{On the hardware side, \libra{} imposes the following requirements:}

\begin{description}
  \item[\hypertarget{hreq:contract}{\hrcontract}] A leakage contract for control-flow balancing is provided.
  \item[\hypertarget{hreq:slice_granularity}{\hrslice}] The \ac{PC} does not leak at a finer granularity than a slice.
  \item[\hypertarget{hreq:uniform_access_pattern}{\hruniform}]
The instruction memory access pattern does not depend on the outcome of the level-offset branch (implied by \hrslice{}).
  \item[\hypertarget{hreq:safe}{\hrsafe}] The level-offset branch and the level-offset call are safe instructions.
\end{description}

\subsubsection*{On the software-side, \libra{} relies on:}

\begin{description}
  \item[\hypertarget{sreq:branches}{\textbf{SR1}}] A correct identification of secret-dependent regions and functions called from secret-dependent regions.
  \item[\hypertarget{sreq:balanced}{\textbf{SR2}}] A secure balancing according to a weak observer mode as prescribed by the leakage contract.
    In practice, this entails making sure that secrets do not directly flow to unsafe instructions, applying a balancing algorithm (such as the one from~\cite{winderix:2021-compiler}), and providing dummy versions for functions called from secret-dependent regions.
  \item[\hypertarget{sreq:fold}{\textbf{SR3}}] A correct folding of the balanced regions and functions. In \cref{sec:libra-transformation}, we give a folding algorithm.
\end{description}

\section{Formal Semantics}\label{sec:libra-formal-semantics}

\subsection{Language and Semantics}\label{sec:semantics}
\subsubsection{Language} %
The \libra{} folding transformation transforms programs written in a source assembly
language \(\sasm\)\footnote{Following common
  practice~\cite{DBLP:journals/corr/abs-2001-11334}, we denote source objects
  with a \src{blue, sans-serif} font and target objects with a \trg{red, bold}
  font. Objects common to source and target are written with black normal font.}
to a target language \(\tasm\) (\cref{fig:syntax}).

\begin{figure}[h]
  \begin{tabularx}{\linewidth}{ r@{~}r@{~}X }
  \multicolumn{3}{l}{
    (Values) $\val{} \in \Vals \quad$%
    (Registers) $\register{x} \in \Regs \quad$%
    (Labels) $\loc, \loc_{t}, \dots \in \Loc \qquad$%
  }\\
  \(\angles{exp}\)  & \(::=\) & \(\val{}\) \(\vert\) \(\register{x}\) \\
  \(\angles{inst}\) & \(::=\) &
  \inst{\uop}{\(\register{x}\ \angles{exp}\)}  \(\vert\)
  \inst{\bop}{\(\register{x}\ \angles{exp}\ \angles{exp}\)} \(\vert\)
  \inst{store}{\(\angles{exp}\ \angles{exp}\)}\\
  & \(\vert\) &
  \inst{br}{\(\angles{exp}\ \loc_{t}\ \loc_{f}\)} \(\vert\)
  \inst{call}{\(\loc\)} \(\vert\)
  \inst{ret}{} \\
  \(\angles{\srcu{inst}}\) &   \(::=\) & \qinst{\qs}{br}{{\(\angles{exp}\ \ell_{t}\ \ell_{f}\)}} %
                            \(\vert\)   \qinst{\qs}{call}{{\(b\ \ell\ \ell'\)}}
                            \(\vert\)   \(\angles{inst}\) \\
    \(\angles{\trgb{inst}}\) & \(::=\) & \qinst{\qlo}{br}{{\(\angles{exp}\ \val\ \val\ \val\)}} \(\vert\)
    \qinst{\qlo}{call}{{\(b\ \ell\)}} \(\vert\)
                                        \(\angles{inst}\)
\end{tabularx}

  \caption{Syntax of \(\sasm\) and \(\tasm\) instructions where \(\protect\inst{\uop}{} \in \{\protect\inst{neg}{}, \protect\inst{load}{} \dots\}\) and \(\protect\inst{\bop}{} \in \{\protect\inst{add}{}, \protect\inst{mul}{}, \dots\}\) are non-control-flow-altering unary and binary instructions and \(b \in \{\bot, \top\}\) is an immediate boolean operand. A program $\prog{}$ is a partial mapping from locations to instructions and $\prog{\loc}$ denotes the instruction at location \(\loc\).}
  \label{fig:syntax}%
\end{figure}

\paragraph{Source language}
In addition to standard ISA instructions, the source language \(\sasm\) is equipped with additional syntactic constructs to:
\begin{enumerate*}[label={\textit{(\arabic*)}}]
  \item identify secret-dependent branches (\srbranch{}), and
  \item associate functions that can be called in secret-dependent regions with a dummy version (\srbalanced{}).
\end{enumerate*}
These constructs should be seen as information derived from source-level annotations.
Secret-dependent branches, \qinst{\qs}{br}{{\(c\ \ell_{t}\ \ell_{f}\)}}, indicate that the condition \(c\) is secret and inform the \libra{} transformation about secret-dependent regions to fold.
Their semantics is similar to regular conditional branches.
Secret-dependent calls, \qinst{\qs}{call}{{\(b\ \ell\ \ell'\)}}, indicate that the function at address \(\ell'\) is the dummy version of the function at address \(\ell\). If \(b = \top\), the original function \(\ell\) is called, whereas if \(b = \bot\), the dummy function \(\ell'\) is called.
Secret-dependent calls inform the \libra{} transformation of functions to fold with their dummy version.

\paragraph{Target language}
The target language is equipped with a level-offset branch and level-offset call, which
are used to navigate folded regions and whose semantics will be detailed later.

\subsubsection{Configurations}
\emph{Source configurations} are of the form \(\aconf{\mem, \reg, \pc, \retstack}\) where \(\mem: \Vals \to \Vals\) is a memory, mapping addresses to values, \(\reg: \Regs \to \Vals\) is a register map, \(\pc\) is the program counter, %
and \(\retstack\) is a stack of return addresses.
\footnote{For simplicity, our formalization features a stack of return addresses. However, a standard setting with a simple return address register that is correctly saved/restored on the stack would be equivalent, under the assumption that return addresses do not interfere with the rest of the program (\ie{} no return address overwrite, no pointer arithmetic on return address, etc).}
To execute a folded region slice-by-slice, \libra{} keeps track of the number of basic blocks in the currently active level (\(\bbc\)) and the offset of the currently active basic block
(\(\curOff\)) in a \emph{\libra{} context}, denoted \(\ctx = (\bbc, \curOff)\).
The initial \libra{} context is \((1, 0)\).
A \emph{\libra{} configuration} \(\conf{}\) is a tuple \(\conf{\mem, \reg, \pc, \retstack, \ctxstack}\) where \(\aconf{\mem, \reg, \pc, \retstack}\) is a source configuration, and \(\ctxstack\) is a stack of \libra{} contexts.
In the following, we refer to \libra{} configurations simply as configurations.

Note that handling function calls and exceptions in folded regions requires a stack of (at least) two \libra{} contexts. In that setting, \libra{} contexts must be saved and restored by the callee in non-leaf functions.
For simplicity, our formalization allows for a stack of unlimited size.

\subsubsection{Semantics}
The semantics of \libra{}, given by the relation \(\tevalconf{\conf{}}{\conf{}'}{\obs}{}\), defines that the evaluation of an instruction in a configuration \(\conf{}\) produces a configuration \(\conf{}'\) and an observation \(\obs\). The semantics is parameterized by a function \(\obsfunc(\conf{})\), which defines the observation produced in a configuration \(\conf{}\) (and will be instantiated in \cref{sec:policy}). %
We give in \cref{fig:libra-semantics} an excerpt of semantics rules, focusing on the important aspects of \libra{}
\ie{} the update of the program counter and the \libra{} context.
The evaluation of an expression $e$ using a register file~$\reg$ is given by $\eeval{e}{\reg}$ and the evaluation of a non--control-transfer instruction \(inst\) (\eg{} arithmetic, logic, or memory instruction), is given by a relation \(\isaeval{inst}{\aconf{\mem, \reg}}{\aconf{\mem', \reg'}}\). %

The program counter always points to the instruction to be executed (\(\prog{\pc}\)). %
To navigate the folded memory layout, we define a function \(\func{slice\_addr}\), returning the address of the current slice, and a function \(\func{next\_slice}\), returning the address directly following the current slice:
\begin{align*}
\func{slice\_addr}(\pc, \curOff) \mydef~& \pc - \curOff\\
\func{next\_slice}(\pc, \bbc, \curOff) \mydef~& \func{slice\_addr}(\pc, \curOff) + \bbc
\end{align*}

The rule \textsc{pc-update} defines the evaluation of a non-control-transfer instruction. It increments the program counter with the basic block count, effectively jumping to current offset in the next slice.

The rule \textsc{lob-true} defines the evaluation of a level-offset branch \qinst{\qlo}{br}{{\(e\ \offf\ \offt\ \bbc'\)}} when the condition \(e\) evaluates to true. It jumps to the next slice at offset \(\offt\) and sets the new basic block count to \(\bbc'\).
The rule \textsc{lob-false} is analogous and omitted for brevity.

The rule \textsc{lo-call} defines the evaluation of a level-offset call, which calls a function folded with its dummy version, at location \(\ell\). The rule jumps to the first slice of the function and, according to the boolean \(b\), sets the offset to 0 or 1, to execute the original or the dummy function, respectively. It also sets the basic block count to 2, to account for the folding of the original and dummy functions. Finally, it pushes the return address on the return stack. Normal function calls are similar, but push the initial \libra{} context \((1,0)\) to the \libra{} stack.

The rule \textsc{ret} defines the evaluation of a return instruction. It simply jumps to the return address on the top of the return stack and restores the previous \libra{} context. %

\begin{figure}[h]
  \begin{mathpar}
    \centering
    \pcupdatenct{}
    \and
    \lobranchtrue{}
    \and
    \locall{}
    \and
    \return{}
  \end{mathpar}
  \caption{Excerpt of the \libra{} semantics, where
\(\qq \in \{\qlo, \qs, \varepsilon\}\) and
    \mbox{\textnormal{\(\obs{} = \obsfunc(\conf{\mem, \reg, \ctxstack \cdot (\bbc, \curOff)})\)}}. %
  }\label{fig:libra-semantics}
\end{figure}

We additionally equip our source language \(\sasm\) with a source semantics
\(\seval{\obs}\), defined in a standard way and omitted here for brevity.
Finally, we let \(\evalconf{\conf{}}{\conf{}'}{\obs}{n}\) be the \(n\)-step evaluation from a configuration \(\conf{}\) to a configuration \(\conf{}'\), where \(\obs\) is the concatenation of observations produced by individual instructions~\cite{barthe:2018-secure}.

\subsection{Security Policy}\label{sec:policy}
\subsubsection{\libra{} Leakage Model}
Side-channel observations are captured in a leakage contract (\hrcontract{}), which partitions the instruction set into leakage classes and safe/unsafe instructions (\cf{} \cref{sec:libra-contract}).

In order to leverage leakage classes in a standard security criterion~\cite{barthe:2018-secure}, we associate a unique leakage identifier (\id{add}, \id{load}, \id{br}, etc.) to each leakage class. The leakage identifier of an instruction \(inst\) is given by \(\class{inst}\). %
For instance, if additions and subtractions are indistinguishable to an attacker, a possible instantiation of \(\class{}\) is \(\class{\inst{add}{\texttt{x\ x\ x}}} = \class{\inst{sub}{\texttt{x\ x\ x}}} = \id{add}\).

Additionally, the instruction set is partitioned into disjoint sets.
Safe unary instructions (\(\IuSafe{}\)) and
safe binary instructions (\(\IbSafe{}\)), do not expose information about the value of their operands.
Conversely, unsafe unary instructions (\(\IuUnsafe\)), left-unsafe (\(\IlUnsafe{}\)), right-unsafe (\(\IrUnsafe{}\)), and
left-right-unsafe (\(\IlrUnsafe{}\)) instructions expose information about the values of their only, left, right, or both source operands, respectively.

\libra{} leaves freedom to hardware developers regarding the concrete instantiation of leakage classes and safe/unsafe partitioning. It only imposes (\hrsafe{}) that secure branches and level-offset branches do not leak their outcome---\ie{}  \qinst{\(\{\qlo,\qs\}\)}{br}{\(c\ \_\)} \(\in \IuSafe{}\)---and secure calls and level-offset calls do not reveal whether the original function or the dummy function is actually executed---\ie{} \qinst{\(\{\qlo,\qs\}\)}{call}{\(\_{}\)} \(\in \IuSafe{}\).
For our security policy, we additionally require that normal branches and calls leak their outcome, and that control-flow-altering instructions belong in a distinct leakage class from each other and from non-control-flow-altering instructions. Intuitively, this ensures that low-equivalent source executions are slice-synchronized: at each step, their program counters belong to the same slice.

\subsubsection{Weak/Strong Observer Mode}
The \libra{} leakage model is used to instantiate the function \(\obsfunc\), which, as mentioned earlier, is a parameter of the semantics specifying the observation produced when evaluating an instruction.
We define two distinct observer modes (\ie{} instantiations of \(\obsfunc\)) that we will apply to \(\sasm\) and \(\tasm\) programs.

The \emph{weak observer mode} (\(\obsfuncweak\)) exposes all timing and
microarchitectural effects that are independent of the program counter (\ie{}
the \emph{balanceable} leakage). The leakage classes and safe/unsafe
partitioning determine the instantiation of \(\obsfuncweak\), as defined in \cref{fig:leakage-model}.

\begin{figure}[ht]
  \input{semantics/leakage}
  \caption{Definition of \(\obsfuncweak\) according to the \libra{} leakage contract (excerpt). Other rules (\textsc{r-unsafe}, \textsc{lr-unsafe}, etc.) are analogous.}\label{fig:leakage-model}
\end{figure}

The \emph{strong observer mode} (\(\obsfuncstrong\)) includes observations of
the weak mode, plus the observable part of the program counter (\ie{} the
\emph{unbalanceable} leakage), which, from \hrslice{}, does not expose more than the address of the current slice:
\begin{align*}
  \obsfuncstrong(\conf{\mem, \reg, \pc, \retstack, \ctxstack \cdot (\bbc, \curOff)}) =\ & %
                                                                                        \func{slice\_addr}(\pc, \curOff) \cdot\\
  & \obsfuncweak(\conf{\mem, \reg, \pc})
\end{align*}

\subsubsection{Security}
Security is defined with respect to a partition of the initial state (memory and registers) into public and secret regions.
\begin{definition}[Indistinguishability]\label{def:indistinguishability} Two
  states \(\conf{}\), \(\conf{}'\) are indistinguishable, written
  \(\conf{} \indist \conf{}'\), if they agree on the value of their public
  registers and public memory locations.
\end{definition}
We define security as (termination-insensitive) \emph{Observational Non-Interference} (ONI)~\cite{DBLP:conf/sp/GoguenM82a} \wrt{} an observation function \(\obsfunc\):
\begin{definition}[\oni{\obsfunc}]\label{def:oni}
  A program \(\prog{}\), interpreted in a semantics \(\eval{}{}\), is secure under observer mode \(\obsfunc\), written \(\oni{\obsfunc}(\prog{})\) if and only if for any pair of initial configurations
  \(\conf{}_{0}\), \(\conf{}_{0}'\), %
  if \(\conf{}_{0} \indist \conf{}_{0}'\), and %
  \(\evalconf{\conf{}_{0}}{\conf{}_{n}}{\obs}{n}\), then
  \(\evalconf{\conf{}_{0}'}{\conf{}_{n}'}{\text{\obs'}}{n}\) and
  \(\obs = \obs'\).
\end{definition}

Intuitively, the goal of our \libra{} transformation is to transform \(\sasm\) programs that are \(\oni{\obsfuncweak}\), to \(\tasm\) programs that are \(\oni{\obsfuncstrong}\).
In other words, developers should make sure that secrets do not directly flow to insecure instructions and balance secret-dependent branches (\srbalanced{}), while \libra{}---with compiler (\srfold{}) and hardware (\hrslice{}) support---guarantees that the target program is secure with respect to a strong observer that can observe (parts of) the program counter through microarchitectural side-channels.

\subsection{Libra Transformation}\label{sec:libra-transformation}
To automatically support \libra{} (\srfold{}), we define a folding transformation \(\transform\) from \(\sasm\) programs---with annotated secret-dependent branches (\srbranch{}) and dummy functions for functions that can be called in secret-dependent regions (\srbalanced{})---to \(\tasm\) programs.
For clarity, we present the transformation informally, with illustrative examples, and leave the formalization to \cref{app:transformation}.

\subsubsection{Folding secret-dependent regions} %
For each \emph{balanced} secret-dependent region \(S\)---annotated in \(\sasm{}\) programs by a secret-dependent branch \qinst{\qs}{br}{{\(e\ \ell_{t}\ \ell_{f}\)}}---the transformation first computes the level structure \(L_{0}\dots L_{n}\) of the region.
Next, for all levels \(L_{i}\), the transformation rewrites each terminating instruction
\qinst{\(\{\varepsilon,\qs\}\)}{br}{\(e\ \ell_{t}\ \ell_{f}\)} in the level with %
a level-offset branch
\qinst{\qlo}{br}{{\(e\ \offt\ \offf\ \bbc\)}} where \(\bbc\) is the basic block count of the next level (\ie{} \(|L_{i+1}|\)), and \(\offt\) and \(\offf\) are the level offsets corresponding to \(\ell_{t}\) and \(\ell_{f}\), respectively, in the level \(L_{i+1}\).
Finally, for each level of the level structure, the transformation folds the corresponding basic blocks by interleaving their instructions.

\begin{example}[Folding branches]
Consider the balanced secret-dependent region in \cref{lst:transform-br:source} and let \(\bb{En}, \bb{t}, \bb{tt} \dots \bb{Ex}\) be the basic blocks corresponding to labels %
\listinglabel{En}, %
\listinglabel{t}, %
\listinglabel{tt}, %
\(\dots\),
\listinglabel{Ex}.
The compiler first computes the level structure of the region: %
\(L_{0} = \{ \bb{En} \}, L_{1} = \{ \bb{t}, \bb{f} \}, %
L_{2} = \{ \bb{tt}, \bb{tf}, \bb{ft}, \bb{ff} \}, %
L_{3} = \{ \bb{Ex} \}\).
Next, the transformation rewrites the terminating instruction in each level (except \(L_{3}\)) with level-offset branches.
For instance, terminating instructions of \(L_{1}\) are replaced with \qinst{\qlo}{br}{{\(c\ \offt\ \offf\ |L_{2}|\)}} where \(\offt, \offf\) are computed according to the mapping \(\{\text{\listinglabel{tt}} \mapsto 0, \text{\listinglabel{tf}} \mapsto 1, \text{\listinglabel{ft}} \mapsto 2, \text{\listinglabel{ff}} \mapsto 3 \}\).
Finally, the transformation interleaves the basic blocks in each level, giving the program in \cref{lst:transform-br:target}.
\end{example}

\begin{listing}[ht]
  \centering
\setcounter{listing}{\value{lstlisting}} %
\addtocounter{lstlisting}{1} %
\protect\captionof{listing}{Libra transformation (\cref{lst:transform-br:target}) of a balanced secret-dependent branch (\cref{lst:transform-br:source}) where \lstinline[style=myasmstyle]|j Ex| is syntactic sugar for \lstinline[style=myasmstyle]|br 0,Ex,Ex|;
\lstinline[style=myasmstyle]|lo.j| is syntactic sugar for \lstinline[style=myasmstyle]|lo.br 0,0:0:1|; and $i_{\text{\texttt{\normalfont{t}}}}, i_{\text{\texttt{\normalfont{t}}}}', \dots$ are arbitrary non-terminating instructions.}\label{lst:transform-br}
\vspace{-1em}
\begin{minipage}{.40\linewidth}
\begin{lstlisting}[
  style=myasmstyle,
  label={lst:transform-br:source},
]
En: s.br c,t,f
 t:  $i_{\text{t}}$; $i_{\text{t}}'$
     br d,tt,tf
tt:    $i_{\text{tt}}$; j Ex
tf:    $i_{\text{tf}}$; j Ex
 f:  $i_{\text{f}}$; $i_{\text{f}}'$
     br e,ft,ff
ft:    $i_{\text{ft}}$; j Ex
ff:    $i_{\text{ff}}$; j Ex
Ex: [...]
\end{lstlisting}
\subcaption[listing]{}\label{lst:transform-br:source}
\end{minipage}
\hfill
\begin{minipage}{.55\linewidth}
\begin{lstlisting}[
  style=myasmstyle,
  label={lst:transform-br:target},
]
En: lo.br c,0:1:2
L1:  $i_{\text{t}}$; $i_{\text{f}}$
     $i_{\text{t}}'$; $i_{\text{f}}'$
     lo.br d,0:1:4
     lo.br e,2:3:4
L2:    $i_{\text{tt}}$; $i_{\text{tf}}$; $i_{\text{ft}}$; $i_{\text{ff}}$;
       lo.j; lo.j; lo.j; lo.j


Ex: [...]
\end{lstlisting}
\subcaption[listing]{}\label{lst:transform-br:target}
\end{minipage}
\vspace{-1.6em}
\end{listing}

\subsubsection{Folding functions} %

First, the algorithm computes the union of the level structures of the functions (the original and the dummy function) to fold. Then, similarly as for secret-dependent branches, it replaces branches with level-offset branches, and interleaves instructions according to the level structure. %
Finally, it replaces the call with a level-offset call \qinst{\qlo}{call}{{\(b\ \ell\)}}, where \(\ell\) is the (fresh) label of the folded function.

\begin{example}[Folding functions]\label{ex:libra-transformation-functions}
Consider the program in \cref{ex:sasm_transform_func} and let \(\bb{foo}, \bb{t}, \bb{f} \dots \bb{Ex'}\) be the basic blocks corresponding to labels %
\listinglabel{foo}, %
\listinglabel{t}, %
\listinglabel{f}, %
\(\dots\),
\listinglabel{Ex'}.
 The compiler first computes the union of the level structure of the functions: %
\(L_{0} = \{ \bb{foo}, \bb{foo'} \}, L_{1} = \{ \bb{t}, \bb{f}, \bb{t'}, \bb{f'} \},
L_{2} = \{ \bb{Ex}, \bb{Ex'} \}\).
The transformation then rewrites the terminating instructions and interleaves the basic blocks in each level, giving the program in \cref{ex:sasm_transform_func_final}.

\begin{listing}[ht]
  \centering
  \centering
\setcounter{listing}{\value{lstlisting}} %
\addtocounter{lstlisting}{1} %
\protect\captionof{listing}{\libra{} transformation (\cref{ex:sasm_transform_func_final}) of a call inside a balanced secret dependent region (\cref{ex:sasm_transform_func}) where \lstinline[style=myasmstyle]|j Ex| is syntactic sugar for \lstinline[style=myasmstyle]|br 0,Ex,Ex|;
\lstinline[style=myasmstyle]|lo.j n| is syntactic sugar for \lstinline[style=myasmstyle]|lo.br n,0:1:2|; and $i_{0}, i_{0}', \dots$ are arbitrary non-terminating instructions.}
\vspace{-1em}
\begin{minipage}{.48\linewidth}
\begin{lstlisting}[
  style=myasmstyle,
  label={ex:sasm_transform_func}
]
[...]
s.call $\top$,foo,foo'
[...]
foo: $i_{0}$
     br c,t,f
  t:   $i_{1}$; $i_{2}$; j Ex
  f:   $i_{3}$; $i_{4}$; j Ex
 Ex: ret
foo': $i_{0}'$
      br c',t',f'
 t':    $i_{1}'$; $i_{2}'$; j Ex'
 f':    $i_{3}'$; $i_{4}'$; j Ex'
 Ex': ret
\end{lstlisting}
\subcaption[listing]{}\label{ex:sasm_transform_func}
\end{minipage}
\hfill
\begin{minipage}{.5\linewidth}
\begin{lstlisting}[
  style=myasmstyle,
  label={ex:sasm_transform_func_final},
]
[...]
lo.call $\top$,ffoo
[...]
ffoo: $i_{0}$; $i_{0}'$
      lo.br c 0:1:4;
      lo.br c' 2:3:4
  L2:   $i_{1}$; $i_{3}$; $i_{1}'$; $i_{3}'$
        $i_{2}$; $i_{4}$; $i_{2}'$; $i_{4}'$;
        lo.j 0; lo.j 0;
        lo.j 1; lo.j 1


Ex'': ret; ret
\end{lstlisting}
\subcaption[listing]{}\label{ex:sasm_transform_func_final}
\end{minipage}
\vspace{-1.6em}
\end{listing}
\end{example}

\subsection{Correctness and Security}\label{sec:correctness-security} %
This section states the correctness and security of our Libra transformation \(\transform\).
First, we establish a correspondence relation between source and target configurations. Intuitively, this relates source and program configurations that are at the same point of execution and have the same memory and register states.

\begin{definition}[\(\sconf{} \protect\libraeq{\sprog{}} \tconf{}\)]\label{def:libraeq}
  A source configuration \(\sconf{} = \conf{\srcu{\mem},\srcu{\reg}, \srcu{\pc}, \srcu{\retstack}}\) for a program \(\sprog{}\) is related to a target configuration \(\tconf{} = \conf{\trgb{\mem}, \trgb{\reg}, \trgb{\pc}, \trgb{\retstack}, \trgb{\ctxstack}}\) for a program \(\tprog{}\), denoted \(\sconf{} \libraeq{\sprog{}} \tconf{}\), if and only if the following holds:
  \begin{enumerate*}
    \item \(\srcu{\mem} = \trgb{\mem}\),
    \item \(\srcu{\reg} = \trgb{\reg}\), and
    \item \(\srcu{\pc} \corrloc{\sprog{}} \trgb{\pc}\),
  \end{enumerate*}
  where \(\corrloc{\sprog{}}\) relates program locations in the source program to their corresponding location in the target program.
\end{definition}

Libra is a correct program transformation, preserving program semantics, as established by the following proposition:
\begin{restatable}[Correctness]{proposition}{correctness}\label{thm:correctness}
  For any \(\sasm\) program \(\sprog{}\), number of steps \(n\), and initial source and target
  configurations \(\sconf{}\) and \(\tconf{}\) such that
  \(\tconf{} \libraeq{\sprog{}} \sconf{}\),
if    \(\sevalconf{\sconf{}}{\sconf{}'}{}{n}\)
then
\(\tevalconf{\tconf{}}{\tconf{}'}{}{n}\)
and
\(\sconf{}' \libraeq{\sprog{}} \tconf{}'\),
  where \(\seval{}\) is parameterized by \(\sprog{}\) and \(\teval{}\) is parameterized by \(\transform(\sprog{})\).
\end{restatable}
\noindent{}

Libra is a program transformation that hardens programs secure against a weak attacker, to programs secure against a strong attacker, as established by the following proposition:
\begin{restatable}[Security]{proposition}{security}\label{thm:security}
  For any \(\sasm\) program \(\sprog{}\),
  \begin{equation*}
    \oni{\obsfuncweak}(\sprog{}) \implies
    \oni{\obsfuncstrong}(\transform(\sprog{}))
  \end{equation*}
\end{restatable}

\noindent{}
\iftechreport
Proof sketches for \cref{thm:correctness,thm:security} are given in \cref{app:proofs}.
\fi

\section{CFL Characterization}\label{sec:characterization}

Based on a rigorous analysis of the microarchitectural attack literature (cf. Table~\ref{tab:cfl} for references), we now present a characterization of hardware optimizations regarding how they have been exploited to leak the control flow of applications.
The importance of this characterization is twofold.
First, it provides a mental framework for improving the understanding of \ac{CFL}, which also guided the design of \libra{}.
Second, it provides the basis to establish recommendations for hardware designers wishing to adopt \libra{}.
The results of our \ac{CFL} attack analysis, \ie{} the raw data for our characterization, are presented in Table~\ref{tab:cfl}.
Each row in this table corresponds to a microarchitectural optimization.
The first column names the optimization
and points to representative papers exploiting it for \ac{CFL} attacks.
The second column indicates if the hardware optimization yields balanceable observations (\ie{} if they can be balanced without \libra{} support).
The third column lists our recommendation on how to handle the leakage using \libra{}.
The last column contains additional notes.

We start by dividing the optimizations into two top-level classes: those that yield balanceable observations (class C1), and those that yield unbalanceable observations (class C2).

\iftrue{}

\newcommand\guidelinehref[1]{\hyperref[guide:c#1]{\textcolor{black}{\textbf{C#1}}}}

\begin{table*}
\caption{Control-flow leakage attack landscape.}
\vspace{-1.0em}
\resizebox{0.9\textwidth}{!}{
\small
\begin{tabular}{ L{4.80cm} C{1.60cm} C{1.60cm} L{7.80cm} }
\toprule
\textbf{Exploited Optimization} & \textbf{Balanceable} & \textbf{Guideline} & \textbf{Notes}\\
\midrule
Computation simplification \cite{andrysco:2015-on} & \cmark & \guidelinehref{1} &
Alternatives: reject program, DIT~\cite{arm:2021-dit, intel:2022-doit} \\
Data TLB \cite{gras:2018-translation, wang:2017-leaky} & \cmark & \guidelinehref{1} & Balance address operands (page granular) \\
Data cache \cite{gruss:2016-flush, liu:2015-last, osvik:2006-cache, percival:2005-cache, yarom:2014-flush} & \cmark & \guidelinehref{1} & Balance address operands (cache-line granular) \\
Data cache bank \cite{yarom:2017-cachebleed} & \cmark & \guidelinehref{1} & Balance address operands (byte granular) \\
DRAM row buffer (data) \cite{pessl:2016-drama} & \cmark & \guidelinehref{1} & Balance address operands \\
Data-dependent data prefetcher \cite{sanchez:2021-opening,chen:2024-gofetch} & \cmark & \guidelinehref{1} &
Balance address operands (loads/stores) and value operands (stores) \\
Load/store buffers & \cmark & \guidelinehref{1} & \\
Pipeline interlock
\cite{sullivan:2018-microarchitectural,yarom:2017-cachebleed,moghimi:2019-memjam} & \cmark & \guidelinehref{1} & Balance stalling data dependencies \\
µop fusion
\cite{ren:2021-i} & \cmark & \guidelinehref{1} & \\
Execution engine  \cite{aldaya:2019-port, bhattacharyya:2019-smotherspectre, gast:2023-squip, gast:2023-remote, rokicki:2022-port, wang:2006-covert} & \cmark & \guidelinehref{1} & Balance structural dependencies \\ %
Interrupt controller \cite{vanbulck:2018-nemesis, moghimi:2020-copycat} & \cmark & \guidelinehref{1} & Balance interrupt latencies \\ %
Reorder buffer (ROB) \cite{aimoniotis:2021-reorder} & \cmark & \guidelinehref{1} & Balance instruction types \\ %
Memory bus / controller \cite{bognar:2022-mind, bognar:2023-microprofiler, wang:2014-timing, wu:2014-whispers} & \cmark & \guidelinehref{1} & Balance memory bus(es) usage \\
Computation reuse \cite{sanchez:2021-opening} & \cmark & \guidelinehref{1} & Balance operands \\
Branch order buffer (BOB) \cite{intel:2017-performance} & \cmark & \guidelinehref{1} & \\
Interconnect \cite{paccagnella:2021-lord} & \cmark & \guidelinehref{1} &  \\
Frontend \cite{puddu:2021-frontal} & \cmark & \hrslice{} & Slice-granular fetch/decode \\ %
Instruction cache \cite{aciicmez:2007-yet, chen:2024-prefetchx, hahnel:2017-high, liu:2015-last, yarom:2014-flush} & \cmark & \guidelinehref{2.2.2} & Balancing (confining region inside a single cache line) is more limited \\
MMU / Page tables \cite{moghimi:2020-copycat, vanbulck:2017-telling, wang:2017-leaky, xu:2015-controlled} & \cmark & \guidelinehref{2.2.2} & Balancing (confining region inside a single page) is more limited \\
DRAM row buffer (instructions) \cite{pessl:2016-drama} &  & \guidelinehref{2.2.2} & \\
Instruction prefetcher \cite{lipp:2022-amd, zhang:2023-bunnyhop} &  & \guidelinehref{2.2.2} & \\
Instruction TLB \cite{gras:2018-translation, wang:2017-leaky} &  & \guidelinehref{2.2.2} & \\
\ac{PC}-dep data prefetcher \cite{chen:2023-afterimage, bhattacharya:2012-hardware, chen:2024-prefetchx, gruss:2016-prefetch,shin:2018-unveiling} &  & \guidelinehref{2.2.3} & \\
Directional predictor \cite{aciicmez:2007-on, aciicmez:2006-predicting, evtyushkin:2018-branchscope, huo:2019-bluethunder} &  & \guidelinehref{2.2.3} & Only for public branches, disable for \lob{} \\
BTB \cite{evtyushkin:2016-jump, lee:2017-inferring, yu:2023-all} &  & \guidelinehref{2.2.3} & Care must be taken not to leak the target transiently \\
Value prediction \cite{deng:2021-new, sanchez:2021-opening} &  & \guidelinehref{2.2.3} & \\
µop cache (DSB) \cite{deng:2022-leaky, kim:2021-uc-check, ren:2021-i} &  & \guidelinehref{2.2.4} &
Alternative: disable in folded regions \\
Silent stores \cite{sanchez:2021-opening} &  & \guidelinehref{2.1} & Disable in folded regions \\
Instruction cache bank \cite{yarom:2017-cachebleed} &  & Disable & Disable in folded regions (violates \hrslice{}) \\
\bottomrule
\end{tabular}}
\label{tab:cfl}
\end{table*}
\fi

\subsection*{C1 - Balanceable observations}
\label{guide:c1}

For optimizations yielding balanceable observations, the hardware can rely on the software to balance these observations according to the \libra{} leakage contract (\srbalanced).

\subsection*{C2 - Unbalanceable observations}
\label{guide:c2}

Optimizations yielding unbalanceable observations inevitably leak the control flow when the processor executes weakly balanced code.
One of objectives of \libra{} is to keep the optimizations in this category enabled as much as possible.
We further break down this category into two subcategories.

\subsubsection*{C2.1 - Inhibiting dummy composition}
\label{guide:c2.1}

The first subcategory groups optimizations that inhibit the composition of a dummy instruction.
Consider for example the silent-store optimization~\cite{lepak:2000-value,sanchez:2021-opening}.
A silent store writes a value to memory that is already present at the specified address.
A silent-store optimization skips writes to memory for silent stores.
This behavior turns a store instruction from a right-unsafe into a full-unsafe instruction since its timing and resource usage will depend not only on the memory address operand, as before, but also on the value to store.
To securely balance an unsafe instruction, both of its operands must be balanced as well.
Yet, since a store affects architectural state, a silent store is the only possible dummy instruction to balance a store, which would leak the control flow in the presence of a silent-store optimization.

\textbf{Guideline:} Disable instances from this optimization class in folded regions.
In case that the composition of a dummy instruction is inhibited by the combination of multiple optimizations, it sometimes suffices to disable only one of them.
An alternative solution to disabling the optimization is to blocklist the affected instruction(s) in the hardware-software security contract (\hrcontract).

\subsubsection*{C2.2 - Observations as a function of the instruction address}
\label{guide:c2.2}

The second subcategory concerns optimizations yielding observations that are a function of the instruction address.
We further divide this subcategory into four optimization classes.

\paragraph{C2.2.1 - Observations that reveal the level offset}
\label{guide:c2.2.1}

Some optimizations yield observations that are inherently different for each instruction within a slice. Hence, they \emph{inevitably} reveal the level offset of an executed instruction.
Take the branch predictor for instance.
The possible targets of a \lob{} instruction are different for each \lob{} of the same slice. Hence, if \lob{} targets are encoded in the branch predictor, an attacker sharing the predictor state could distinguish \lob{} instructions within a slice and learn the level offset.

\textbf{Guideline:} Disable optimizations of this type in folded regions.
For some optimizations, it is necessary to completely disable them (e.g., cache banking~\cite{yarom:2017-cachebleed}), for others this might be unnecessary, such as in the example of the directional predictor we gave, which must only be disabled for a \lob{} instruction.

\paragraph{C2.2.2 - Libra-safe optimizations}
\label{guide:c2.2.2}
Some optimizations, such as the instruction cache, directly benefit from \hruniform{}, which imposes a data-oblivious access pattern to the instruction memory. If the processor frontend follows \hrslice{}, by implementing slice-granular fetch/decode, these optimizations do not leak at a finer granularity than a slice.

\textbf{Guideline:} No hardware modifications are required.

\paragraph{C2.2.3 - PC-dependent mappings}
\label{guide:c2.2.3}

Some optimizations map instruction addresses to instruction-specific information.
The BTB and the PC-based strided data prefetcher, for instance, are typically implemented using table-based structures indexed by instruction address.

\textbf{Guideline:}
Thanks to folding, it becomes possible to represent instruction-specific information as slice-specific information.
Mappings from instruction address to instruction-granular information can be changed into mappings from slice address to slice-granular information (per \hrslice).
This usually requires minimal hardware modifications such as indexing hardware structures by slice address instead of by instruction address.

\paragraph{C2.2.4 - Instruction-specific optimizations}
\label{guide:c2.2.4}

Some optimizations perform different operations depending on the instruction and are not generalizable to the slice, contrary to optimization class \guidelinehref{2.2.3}.
An example is the µop cache, where the
operations, \emph{decode}, \emph{insert} and \emph{evict}, depend on the specific instruction.

\textbf{Guideline:} Instead of disabling these optimizations, it might be more beneficial to keep them enabled and always perform the operation on every instruction of the slice.
Keeping optimizations enabled for instances for this optimization class will typically be more expensive compared to optimization class \guidelinehref{2.2.3}.

\section{Implementation and Evaluation}

\subsection{Implementation}\label{sec:implementation}

Following the requirements from \cref{sec:design-choices} and the guidelines from \cref{sec:characterization}, we implemented \libra{} on Proteus~\cite{bognar:2023-proteus} (version 2024.01-O), a RISC-V out-of-order core designed to experiment with hardware security extensions.

\paragraph{\protect\hrcontract{}: Leakage contract}

We partitioned the RISC-V instruction set into leakage classes and validated the correctness of this classification via our automated security evaluation (cf. Subsection~\ref{sec:evaluation}).
In particular, load instructions leak their address via the data cache and are balanced in software.

\paragraph{\protect\hrslice{}: Slice-granular PC leakage}

Based on our analysis, the sources of precise PC leakage on Proteus were the branch target predictor, the instruction cache, and the instruction prefetcher.
For security reasons, we completely disable the branch predictor in folded regions.
Yet, thanks to the linear layout of folded regions, the performance impact of this is limited:
the next slice---where the execution will continue---will be prefetched by the time the branch condition is resolved.
The other hardware structures did not have to be altered, as explained next.

\paragraph{\protect\hruniform{}: Data-oblivious instruction memory access pattern}

The instruction fetch unit has been made \libra{}-aware.
In a secret-dependent folded region, the level offset of the currently executing instruction needs to be invisible to the memory subsystem.
In our implementation, this is achieved by fetching in a fixed order all cache lines including instructions from the current slice.
This also results in the state of the instruction cache being independent from the level offset.
As the prefetcher's behavior in Proteus only depends on the instruction cache state, it also observes the same access patterns and does not require any additional changes.
Thanks to the folded layout in memory, the prefetching remains very effective during the execution of folded regions.

\paragraph{\protect\hrsafe{}: Level-offset branch}

For our implementation, we introduced a variant of the \lob{} instruction, the \emph{terminating level-offset branch} \tlob{}.
This instruction behaves similarly as a regular \lob{}, but additionally encodes the number of slices in the next level, an optimization that makes the \lob{} instructions of the last level of a folded region unnecessary.
Our prototype encodes the \lob{} and \tlob{} variants for each RISC-V branch instruction by repurposing the two prefix bits in the fixed-width 32-bit RISC-V encoding, but other implementations could use the free opcode slots as defined by the RISC-V specification.
The current and previous \libra{} contexts are stored in a two-level hardware stack. We support folded regions with up to 16 basic blocks per level (8 for a terminating level).

\subsection{Evaluation}\label{sec:evaluation}

We evaluated our implementation by measuring the binary size and execution time overheads using a benchmark suite from related work~\cite{bognar:2023-microprofiler,pouyanrad:2020-scf,tsoupidi:2023-thwarting,winderix:2021-compiler,winderix:2024-architectural}; measuring the hardware overhead; and performing RTL-level noninterference testing to validate security.

\paragraph{Binary size.}

The results on binary size can be found in Table~\ref{tab:size-overhead}, which shows the binary size of the original benchmark, the overhead of balancing the secret-dependent branches (which still leaks information through unbalanceable observations), the overhead of linearizing the secret-dependent branches with Molnar's method~\cite{molnar:2005-program}, and finally, folding the secret-dependent branches with \libra{}.
The benchmarks show that the overhead is small compared to state-of-the-art linearized (constant-time) code.
In certain cases, the folded program can even be expressed more succinctly due to the characteristics of folded regions; after the last slice, the next instruction will be executed regardless of which branch was taken, making additional jump instructions, such as in Listing~\ref{running-vulnerable}, unnecessary.

\begin{table}[t]
  \footnotesize
  \caption{Overhead factors: execution time (cycles) / binary size (bytes).}
  \vspace{-.8em}
  \label{tab:size-overhead}
  \resizebox*{\linewidth}{!}{
    \begin{tabular}{l r c c c}
      \textbf{Benchmark} & \textbf{Baseline} & \textbf{Balanced} & \textbf{Linearized} & \textbf{Folded} \\

      \midrule

      fork           & 110 c   / 136 B  & 1.00x / 1.00x              & 1.11x / 1.12x               & 1.00x / 0.94x    \\
      triangle       & 116 c   / 132 B  & 1.03x / 1.06x              & 1.05x / 1.15x               & 0.98x / 1.00x    \\
      bsl            & 1415 c  / 336 B  & 1.20x / 1.04x              & 1.54x / 1.08x               & 1.24x / 1.01x    \\
      diamond        & 186 c   / 192 B  & 1.07x / 1.10x              & 1.18x / 1.23x               & 1.06x / 1.04x    \\
      kruskal        & 1573 c  / 452 B  & 1.09x / 1.05x              & 1.21x / 1.16x               & 1.16x / 1.04x    \\
      ifthenloop     & 407 c   / 200 B  & 1.35x / 1.20x              & 1.28x / 1.20x               & 1.56x / 1.16x    \\
      switch         & 1402 c  / 500 B  & 2.11x / 1.41x              & 2.70x / 1.92x               & 1.90x / 1.15x    \\
      sharevalue     & 1410 c  / 500 B  & 1.38x / 1.02x              & 1.76x / 1.15x               & 1.77x / 1.01x    \\
      mulmod16       & 339 c   / 276 B  & 1.23x / 1.01x              & 1.47x / 1.16x               & 1.32x / 0.96x    \\
      keypad         & 3490 c  / 416 B  & 2.86x / 1.08x              & 3.48x / 1.12x               & 3.61x / 1.06x    \\
      modexp2        & 11716 c / 324 B  & 1.72x / 1.02x              & 1.79x / 1.09x               & 1.78x / 1.01x    \\

      \midrule

      \textbf{mean}  &               & \textbf{1.38x / 1.09x}   & \textbf{1.57x / 1.20x}    & \textbf{1.46x / 1.03x}  \\

    \end{tabular}
  }
\vspace{-2.0em}
\end{table}

\paragraph{Execution time.}

We evaluate the execution time overhead using the same extended benchmark suite, shown in Table~\ref{tab:size-overhead}.
Even though our prototype implementation is not optimal, the benchmarks clearly show an advantage of \libra{} over linearized code.
The mean performance overhead of \libra{} is 46\% compared to 57\% of the linearized code (a relative overhead reduction of 19.3\%), and for certain benchmarks it not only performs much better than linearized code, but also outperforms insecure balanced code.
For example, the \texttt{switch} and \texttt{triangle} benchmarks clearly show the power of \libra{} over alternative approaches.

\paragraph{Hardware cost.}

We evaluate the hardware cost of implementing \libra{} on Proteus by synthesizing the design to the Xilinx  XC7A35TICSG324-1L FPGA in Xilinx Vivado 2022.2.
According to our measurements, the Libra additions increase the number of look-up tables by 11.4\% (from 16,531 to 18,414), the number of registers by 9.5\% (from 13,566 to 14,850), while keeping the critical path unchanged (37.4 ns).

\paragraph{Security.}

To evaluate the security of our implementation, we adopt a methodology from related work~\cite{bognar:2023-microprofiler,winderix:2021-compiler,winderix:2024-architectural}: noninterference-based testing in a cycle-accurate Verilog simulator.
For each benchmark, we manually ensure that all possible code paths are explored, which is feasible due to the relatively small size of the benchmarks.
We verify that, for executions with identical public inputs but varying secret inputs, the processor's internal signals associated with side channels remain consistent.
Any variation would indicate a leak of secret information.
The signals we focus on include the state of the branch predictor, addresses in the instruction and data caches, the state of the instruction prefetcher, and the occupancy of the execution units.
Each simulation is run independently, starting from a cold microarchitectural state.

Interestingly, our security evaluation revealed that the hardened \texttt{kruskal} benchmark (originally introduced in~\cite{mantel15transforming}) contains a recursive function with a secret-dependent number of iterations, as hypothesized by the original authors.
As a linearized implementation of Kruskal's algorithm is not a trivial effort and out of scope for our paper, we only transformed the secret-dependent branch in the main function of the benchmark.

\section{Discussion}
\label{sec:future}

\paragraph{Intra-cache-line attacks}

\hrslice{} requires that the \ac{PC} does not leak at a finer granularity than a level slice.
This implies that executing folded regions is secure only if an attacker is unable to observe intra-cache-line instruction memory accesses in folded regions.
To the best of our knowledge, only two published attacks expose intra-cache-line accesses: cache-bank conflicts~\cite{moghimi:2019-memjam} and false dependencies~\cite{yarom:2017-cachebleed}.
To comply with \hrslice{}, the optimizations exploited by these attacks must be disabled in folded regions. However, we do not consider them to be performance-critical.
In more recent microarchitectures these leakages have been closed, confirming our assumption that \hrslice{} will not significantly affect the performance. %

\paragraph{Future work}

There are some open questions that should be addressed in future work.
First, a limitation with the implementation of our prototype is that the pipeline stalls after fetching a \lob{} (until its condition is resolved).
As described in Section~\ref{sec:libra-overview}, the linearity of folded regions removes the uncertainty of what instructions to \emph{fetch} after a \lob{}.
However, the uncertainty of what instruction to \emph{execute} (i.e., what is the level offset of the next instruction in the next prefetched slice?) still remains.
For security reasons, the CPU cannot proceed based on a prediction of the direction of the \lob{} as this would induce a timing signal exposing the control flow.
On the RISC-V processor we used for our implementation, the \lob{} penalty is generally only a few cycles.
However, the penalty on superscalar CPUs with deeper pipelines (capable of fetching and executing multiple instructions in a single cycle) is much higher.
How to deal with the \lob{} penalty on this class of CPUs (up to 10-15 cycles on some CPUs~\cite{hennessy:2011-computer}) remains an open design question.
We believe that exploiting the regularity and the linearity of folded regions is key to solving this challenge.

Second, we informally argue (\cref{tab:cfl}) that many optimizations either comply with \hrslice{}, or can be adapted to do so.
Our empirical evaluation on an implementation featuring instruction and data caches, branch predictor and prefetcher, supports this argument.
Formally verifying that these optimizations comply with \hrslice{} would be an interesting avenue for future work.

Third, there is no compiler support for \libra{}. We manually identify, balance and fold secret-dependent regions at assembly level, restricting the size of our benchmark programs.
Compiler support for \libra{} to be able to conduct more extensive performance measurements on real-world programs is future work.

Fourth, we put software-based fault attacks~\cite{murdock:2020-plundervolt} and software-based power attacks~\cite{lipp:2021-platypus} out of scope. It is interesting future work to extend the leakage contract to cover these attacks.

Fifth, we only considered fixed-length instructions.
Variable-length instructions, as found in the popular x86 ISA, might require some changes to the current \libra{} design.
To increase the chances of adoption on architectures with variable-length instructions, it would be interesting to investigate how to support them.

Finally, creating the leakage contract is a manual effort.
A very interesting avenue for future work is to investigate how to generate the leakage contract from the hardware description (RTL level), as done in recent work~\cite{dinesh:2024-conjunct, hsiao:2022-scalable, mohr:2024-synthesizing}, and how to express contracts such that they can be consumed by a compiler~\cite{DBLP:journals/corr/abs-2012-14205}.

\section{Related Work}

\paragraph{CFL Hardening}

The literature contains a vast amount of prior work on software-only countermeasures against \ac{CFL} attacks.
Almost 25 years ago, Agat~\cite{agat:2000-transforming} already proposed a transformational security type system to balance conditional branches, later refined by Köpf and Mantel~\cite{kopf:2007-transformational}.
Non-transformational type systems to detect unbalanced branches have been implemented for the AVR~\cite{dewald:2017-avr} and MSP430~\cite{pouyanrad:2020-scf} architectures.
Winderix et al.~\cite{winderix:2021-compiler} proposed an algorithm for control-flow balancing during compilation and implemented and evaluated it using the LLVM compiler infrastructure~\cite{lattner:2004-llvm}.
Prior work considers secret-dependent control flow inherently insecure and strongly discourages control-flow balancing as a defense against \ac{CFL} attacks.
This view is incorporated in the well-established constant-time programming discipline~\cite{almeida:2016-verifying}, which disallows programmers from writing secret-dependent branches.
There is a rich literature to automatically detect~\cite{geimer:2023-systematic, jancar:2021-they} and eliminate~\cite{borello:2021-constantine,molnar:2005-program,rane:2015-raccoon,soares:2023-side,vanoverloop:2024-compiler,wu:2018-eliminating} secret-dependent control flow.

\paragraph{Architectural Support}

Many existing software-only countermeasures
leverage hardware primitives designed with performance in mind.
The resulting security guarantees are brittle, as these countermeasures rely on undocumented behavior that is not guaranteed in future versions of the hardware.
For instance, conditional execution (\aka predicated execution or predication) is supported in some form by the x86, Arm and RISC-V ISAs to accelerate some hard-to-predict branches, yet is sometimes used to eliminate secret-dependent control flow~\cite{borello:2021-constantine,coppens:2009-practical,wu:2018-eliminating}, critically relying on the (current) data-oblivious behavior.
Another example is Intel TSX, used as a primitive for a countermeasure proposed by Gruss et al.~\cite{gruss:2017-strong}.
In contrast, \libra{} provides principled support by augmenting the ISA with a security contract representing its security guarantees.
Many modern CPUs provide safe instructions, making explicit security guarantees part of the ISA.
An example is constant-time support for AES to improve speed and security of applications relying on it (e.g.,~\cite{akdemir:2010-breakthrough}).
As another example, x86, Arm  and RISC-V ISAs have extended their ISAs with facilities to turn unsafe instructions into safe instruction via a feature called Data (Operand) Independent Timing~\cite{arm:2021-dit, intel:2022-doit}.
To the best of our knowledge, architectural support to securely execute balanced code on high-end processors has not been proposed before.
Recently, Winderix et al.~\cite{winderix:2024-architectural} proposed architectural support for control-flow balancing and control-flow linearization to efficiently counter \ac{CFL} attacks.
Unfortunately, their solution for control-flow balancing is only targeted towards processors with a microcontroller profile featuring simple processor pipelines.
In contrast, our proposal for control-flow balancing is designed to securely execute balanced regions on high-end systems.

\paragraph{Hardware-Software Leakage Contracts}

Recent work on hardware-software leakage contracts~\cite{guarnieri:2021-hardware, heiser:2018-safety, lowe:2018-position, mosier:2022-axiomatic, ponce:2021-cats} proposes to augment the ISA with a specification of how the hardware leaks information.
\libra{} also specifies such a hardware-software leakage contract and partitions the instruction set into leakage classes.
Instructions of the same leakage class leak the same information and thus are indistinguishable to an attacker.
The first idea for classifying the instruction set this way was proposed by Winderix et al.~\cite{winderix:2021-compiler} under the form of \emph{latency classes}, a concept that was later generalized into \emph{leakage classes} by Bognar et al.~\cite{bognar:2023-microprofiler}.
Yu et al.~\cite{yu:2019-data}, propose ISA design principles for data Oblivious ISAs (OISAs) to perform side-channel resistant and high-performance computations.
The authors proposes an ISA-level data oblivious abstraction, which partitions the instruction set into safe and unsafe instructions.
In contrast to \libra{}, their work does not include ISA-level principles to make control-transfer instructions data oblivious, and hence is complementary to ours.

\paragraph{Secure Compilation for Side-Channel Defenses}
Our secure compilation proof is inspired by existing proof techniques for preservation of side channel defenses by compilers~\cite{DBLP:journals/pacmpl/BartheBGHLPT20, barthe:2018-secure,DBLP:conf/csfw/BartheBHP21} with some adaptations to account for important differences. Compared to the constant-time policy~\cite{DBLP:journals/pacmpl/BartheBGHLPT20}, balancing allows program counters of low-equivalent executions to diverge, and compared to constant-resource transformation~\cite{DBLP:conf/csfw/BartheBHP21}, our transformation is not entirely leakage preserving.
In this work, we assume a non-canceling leakage model (\ie{} \(\obs_{1} \cdot \obs_{2} =  \obs_{1}' \cdot \obs_{2}' \implies \obs_{1} = \obs_{1}' \wedge \obs_{2} = \obs_{2}'\)).
Secure compilation for relaxed policies based on canceling leakage (\eg{} program cost in terms of clock cycles) have been proposed~\cite{DBLP:conf/csfw/BartheBHP21}. However, it remains unclear whether there exist a concrete threat model (attacker model and microarchitecture) to which these policies securely apply.
In particular, such relaxed policies are insecure against the strong attacker that we consider in this paper~\cite{vanbulck:2018-nemesis}.

\section{Conclusion}

In this paper, we challenged the widely-held belief that control-flow balancing is either insecure or inefficient on modern out-of-order CPUs.
We proposed \libra, a novel hardware-software codesign for principled, secure and efficient balanced execution.
We gave evidence that it is possible with minimal hardware support to securely balance secret-dependent control flow while keeping performance-critical hardware optimizations enabled.
A key feature of \libra{} is the specification of a hardware-software security contract that software can rely on to harden applications in a principled way, similar to how software relies on an ISA specification for the functional correctness of programs.
\libra{} minimally extends the instruction set to make balanced execution secure and efficient on high-end systems, mainly by virtue of folding.
We formalized the \libra{} semantics and the folding transformation, which we proved correct and secure.
We also presented a characterization of how microarchitectural optimizations can leak a program's control flow, the basis for our recommendations for hardware designers wanting to adopt \libra{} to their designs.
Our implementation and evaluation show significant performance benefits compared to state-of-the-art control-flow linearization at low hardware cost.

\begin{acks}
  We thank the anonymous reviewers for their valuable feedback.
  This research was partially funded by the ORSHIN project (Horizon Europe grant agreement \#101070008), the Research Foundation – Flanders (FWO) via grants \#G081322N and \#12B2A24N, and the Flemish Research Programme Cybersecurity.
\end{acks}

\bibliographystyle{ACM-Reference-Format}
\bibliography{bibliography}

\appendix

\iftechreport

\section{Folding Transformation}\label{app:transformation}
This section details our folding transformation \(\transform\) from \(\sasm\) programs to \(\tasm\) programs (\cref{app:algo}). We also discuss the hypotheses we make on source programs (\cref{app:more-properties}).

\paragraph{Notations}
In the following, we let \(B[i]\) denote the \(i\)\textsuperscript{th} instruction in basic block \(B\). We also let \(B_{\varepsilon}\) denote the empty basic block.
By definition of a control-flow graph, all program labels point to the beginning of a basic block and we write \(B_{\loc}\) to denote the basic block corresponding to label \(\loc\).
In a level \(L\), we assume that basic blocks are indexed, and we let \(\bbidx{L, B}\) denote the index of basic block \(B\) in level \(L\).

We let \(S = (B_{entry}, B_{exit}, \{B_{1}, \dots, B_{n}\})\) denote a secret dependent region \(\{B_{1}, \dots, B_{n}\}\), with entry block \(B_{entry}\) and exit block \(B_{exit}\). We let \(\entry{S}\) and \(\exit{S}\) return the entry and exit basic block of the region, respectively.
We let \(F = (\loc_{f}, \{B_{1}, \dots, B_{n}\})\) denote a function, defined by its entry label \(\loc_{f}\) and its sequence of basic blocks \(\{B_{1}, \dots, B_{n}\}\).
Finally, we let \(\func{level\_struct}(S)\) return the level structure of a secret dependent region $S$ (excluding \(B_{exit}\) and \(B_{exit}\)) and \(\func{level\_struct}(F)\) return the level structure of a function.

\subsection{Folding Algorithm}\label{app:algo}
We first define a function \FInsertLob{$L_{i}$, $L_{i+1}$}, which rewrites the terminating instructions from a level \(L_{i}\) (\cref{alg:libra-insert-lob}).
It replace all branches in the level with a level offset branch
\qinst{\qlo}{br}{{\(e\ \offt\ \offf\ \bbc\)}} where \(\bbc\) is the basic block count of the next level (\ie{} \(|L_{i+1}|\)), and \(\offt\) (resp. \(\offf\)) is the basic block number corresponding to \(\ell_{t}\) (resp. \(\ell_{f}\)) in the level \(L_{i+1}\).

\begin{algorithm}
  \protect\caption{Rewriting of terminating instructions.}\label{alg:libra-insert-lob}
\Fn{\FInsertLob{$L_{i}$, $L_{i+1}$}}{
\KwData{Level to modify \(L_{i} = \{B_{0} \dots B_{n}\}\)}
\KwResult{Modified level \(L_{i}' = \{B_{0}' \dots B_{n}'\}\)}
\(\bbc \gets |L_{i+1}|\)\; %
$len \gets |B_{0}|$\;
$\{B_{0}' \dots B_{n}'\} \gets \{B_{0} \dots B_{n}\}$\;
\For{$B_{i} \in L_{i}$}{
  \Assert{} \(|B_{i}| = len\)\;
  \Switch{\(B_{i}[len-1]\)}{
    \Case{\qinst{\(\{\varepsilon,\qs\}\)}{br}{{\(e\ \ell_{t}\ \ell_{f}\)}}}{
  \Assert{} \(B_{\ell_{t}}, B_{\ell_{f}} \in L_{i+1}\)\;
  $\offt \gets \bbidx{L_{i+1}, B_{\ell_{t}}}$\;
  $\offf \gets \bbidx{L_{i+1}, B_{\ell_{f}}}$\;
  $B_{i}'[len-1] \gets %
  \qinst{\qlo}{br}{{$e\ \bbc\ \offt\ \offf$}}
  $\;
    }
    \Case{\inst{ret}{}}{
      $B_{i}'[len-1] \gets \inst{ret}{}$\;
    }
    \Other{\Assert{False}}
  }
}
\Return{\(L_{i}' = \{B_{0}' \dots B_{n}'\}\)}
}
\end{algorithm}

Next, we define a function \FFoldLevel{$L$}, which takes a level \(L = \{B_{0} \dots B_{n}\}\) and returns a single folded basic block \(B'\) interleaving the instructions of \(B_{0} \dots B_{n}\) (\cref{alg:libra-fold-level}).
\begin{algorithm}
  \protect\caption{Folding of the basic blocks of a level $L$.}\label{alg:libra-fold-level}
\Fn{\FFoldLevel{$L$}}{
\KwData{Level to fold \(L = \{B_{0} \dots B_{n}\}\)}
\KwResult{Basic block \(B'\) (folded level)}
$\bbc{} \gets |L|$\;
$len \gets |B_{0}|$\;
$B' \gets B_{\varepsilon}$\;
\For{$B_{i} \in L$}{
  \Assert{} \(|B_{i}| = len\)\;
  \For{$j \in \ocinterval{0,len}$}{
    $B'[j \times \bbc + i] \gets B_{i}[j] $\;
  }
}
\Return{\(B'\)}
}
\end{algorithm}

To fold secret-dependent branches, we define a function \FFoldRegion{$S$} that returns the folded version of a secret dependent region \(S\) (\cref{alg:libra-transformation-branches}). The function first computes the level structure of \(S\), rewrite the terminators, and fold the level structure.

\begin{algorithm}
  \protect\caption{Folding of a secret dependent region $S$.}\label{alg:libra-transformation-branches}
\Fn{\FFoldRegion{$S$}}{
\KwData{Secret-dependent region $S$}
\KwResult{Folded secret-dependent region $S'$}
$L_{1}\dots L_{n} \gets \func{level\_struct}(S)$\;
$L_{0} \gets \entry{S}$\;
$L_{n+1} \gets \exit{S}$\;
\For{$L_{i} \in L_{0} \dots L_{n}$}{
  $L_{i}' \gets \FInsertLob(L_{i}, L_{i+1}, \textit{False})$\;
  $B_{i}' \gets \FFoldLevel(L_{i}')$\;
}
$S' \gets (B_{0}', \exit{S}, \{ B_{1}' \dots B_{n}' \})$\;
\Return{\(S'\)}
}
\end{algorithm}

To fold a function \(F\) with its dummy function \(F'\), we define a function \FFoldFunct{$F, F', \loc_{\textit{ff'}}$} (\cref{alg:fold-functions}) that returns a folded function with entry label $\loc_{\textit{ff'}}$.
First, the algorithm computes the union of the level structure of \(F\) and \(F'\). Then, it replaces branches with level-offset branches. Finally, it folds basic blocks according to the level structure and returns the final function, defined by the set of folded basic blocks and entry label \(\loc_{\textit{ff'}}\).

\begin{algorithm}
  \protect\caption{Folding of a function \(F\) with its dummy version \(F'\). The union of level structures is defined as the componentwise union of basic blocks.}\label{alg:fold-functions}
\Fn{\FFoldFunct{$F,F', \loc_{\textit{ff'}}$}}{
\KwData{Functions $F$, $F$}
\KwResult{Folded function $F''$}
$L_{0}\dots L_{n} \gets \func{level\_struct}(F) \cup \func{level\_struct}(F')$\;
$L_{n+1} \gets B_{\varepsilon}$\;
\For{$L_{i} \in L_{0} \dots L_{n}$}{
  $L_{i}'' \gets \FInsertLob(L_{i}, L_{i+1}, \text{True})$\;
  $B_{i}'' \gets \FFoldLevel(L_{i}'')$\;
}
\Return{\(F'' = (\loc_{\textit{ff'}}, \{B_{0}'' \dots B_{n}'')\}\)}
}
\end{algorithm}

Finally, we define a last function \FInsertLocall{$B$}, which replaces secret dependent calls with their corresponding level-offset call (\cref{alg:rewrite-calls}).

\begin{algorithm}
  \protect\caption{Rewriting of secret dependent calls, where \(\loc_{\textit{ff'}}\) is the label of the folded function corresponding to functions \(\loc_{\textit{f}}\) and \(\loc_{\textit{f'}}\).}\label{alg:rewrite-calls}
\Fn{\FInsertLocall{$B$}}{
\KwData{Basic block to rewrite \(B\)}
\KwResult{Modified basic block \(B'\)}
$B' \gets B$\;
\For{$i \in 0 \dots |B| - 1$}{
  \If{$B[i] = \qinst{\qs}{call}{{\(b\ \ell_{\textit{f}}\ \ell_{\textit{f'}}\)}}$}{
  $B[i] \gets \qinst{\qlo}{call}{{\(b\ \loc_{\textit{ff'}}\)}}$\;
  }
}
\Return{B'}
}
\end{algorithm}

\paragraph{Final folding transformation}
The final folding transformation \(\tprog{} = \transform(\sprog{})\) performs the following steps:
\begin{enumerate}
  \item For each pair of function/dummy $F_{\ell_{f}}/F'_{\ell_{f'}}$ that can
        be called from a secret-dependent region (annotated in \sasm{} with a secret dependent call), \(\transform\) computes the folded function
        \(F_{\loc_{\textit{ff'}}} =\) \FFoldFunct{$F_{\ell_{f}}, F'_{\ell_{f'}}, \loc_{\textit{ff'}}$} and places it at location \(\loc_{\textit{ff'}}\) in \tprog{}. Functions $F_{\ell_{f}}, F'_{\ell_{f'}}$ are also included in \tprog{} if they can be called with ``normal'' calls.
  \item For each (outermost) secret-dependent region \(\srcu{S}\) (annotated in
        \(\sasm{}\) programs by a secret-dependent branch), \(\transform{}\) computes the folded region
        \(\trgb{S} = \) \FFoldRegion{$\srcu{S}$}, and replaces the
        original region $\srcu{S}$ with \(\trgb{S}\). Note that for a given secret-dependent region, our algorithm folds its entire level structure (which includes nested branches). Hence, for nested secret-dependent branches, only the outermost branch need to be considered and the nested branches will be automatically folded.
  \item All other basic blocks are directly copied from \(\sprog{}\) to \(\tprog{}\);
  \item Secret dependent calls in \tprog{} are replaced by \(\FInsertLocall{$B$}\);
  \item Finally, in the final code memory layout of \(\tprog{}\), the folded levels
        are placed adjacent to each other, in level order. In other words the basic block corresponding to \FFoldLevel{$L_{i+1}$} directly follows the basic block corresponding to \FFoldLevel{$L_{i}$}. The exit block of a secret dependent region is also placed just after the last folded level.
\end{enumerate}
Note that from \(\transform{}\), we can naturally define a correspondence from locations in \(\sprog{}\) to locations in \(\tprog{}\), which we capture with the relation \(\sloc{} \corrloc{\sprog{}} \tloc{}\). For folded functions, a source location \(\sloc{}\) in a function $F_{\ell_{f}}$ can be related to two target locations: one in the folded function \(F_{\loc_{\textit{ff'}}}\), and one in the original function $F_{\ell_{f}}$.

\subsection{Hypotheses on Source Programs}\label{app:more-properties}

Our folding transformation (in particular the function \FInsertLob{}) is only defined for source programs satisfying the following requirements:
\begin{enumerate}
  \item In a function, we assume that the terminating instruction of basic blocks in the last level is a return.
  \item In a secret-dependent region or function, all the basic blocks of a given level have the same number of instructions.
  \item In a secret-dependent region or function, all the successors of basic blocks in a level $L_{i}$ are in level $L_{i+1}$. %
\end{enumerate}

Additionally, for correctness and security, we assume that the following hypotheses hold on source programs:
\begin{hypothesis}\label{hyp:safe}
  \(\sprog{}\) is safe. In particular, it implies that starting from an initial configuration, \(\seval{}\) does not get stuck.
\end{hypothesis}

\begin{hypothesis}\label{hyp:cfi}
  Functions in \(\sprog{}\) can only be entered through their entrypoint and always return to their return site.
\end{hypothesis}

\begin{hypothesis}\label{hyp:sese}
  Secret-dependent regions in \(\sprog{}\) are single-entry single-exit: they can only be
  entered through their entry block and are always exited through their exit
  block.
\end{hypothesis}

\section{Proofs}\label{app:proofs}
\newcommand\corr{\corrloc{\sprog{}}}
\newcommand\simulates{\libraeqstrong{\sprog{}}}
\newcommand\corrctx{\CorrectStack{\tprog{}}}

The following lemma establishes a correspondence between instructions of source and target programs when their locations are related by \(\corr\):
\begin{lemma}\label{lemma:t-corr}
  For a location \(\sloc\) in an \(\sasm\) program \(\sprog{}\),
  and location \(\tloc\) in \(\tprog{} = \transform(\sprog{})\) such that
  \(\sloc \corr \tloc\), the following hold:
  \begin{itemize}
    \item iff \(\sprog{\sloc} = \inst{br}{}\ e\ \sloc_{t}\ \sloc_{f}\) and \(\sloc\) is not in a secret-dependent region or function,
    then \(\tprog{\tloc} = \inst{br}{}\ e\ \tloc_{t}\ \tloc_{f}\), with \(\sloc_{t} \corr \tloc_{t}\) and \(\sloc_{f} \corr \tloc_{f}\);
    \item iff \(\sprog{\sloc} = \inst{call}{}\ \sloc_{\mathtt{foo}}\)
    then \(\tprog{\tloc} = \inst{call}{}\ \tloc_{\mathtt{foo}}\), with \(\sloc_{\mathtt{foo}} \corr \tloc_{\mathtt{foo}}\);
    \item if \(\sprog{\sloc} = \qinst{\qs}{br}{}\ e\ \sloc_{t}\ \sloc_{f}\)
    then \(\tprog{\tloc} = \qinst{\qlo}{br}{}\ e\ \bbc\ \offt\ \offf\) where:
    \begin{itemize}
      \item \(\bbc\) is the size of the level \(\lvl\) (let \(\tloc_{\lvl}\) denote its location) containing \(\sloc_{t}\) and \(\sloc_{f}\),
      \item \(\offt\) is an offset such that \(\offt + \tloc_{\lvl} \corr \sloc_{t}\), and
      \item \(\offf\) is an offset such that \(\offf + \tloc_{\lvl} \corr \sloc_{f}\).
    \end{itemize}
    \item if \(\sprog{\sloc} = \inst{br}{}\ e\ \sloc_{t}\ \sloc_{f}\) and \(\sloc\) is in a secret dependent region, then \(\tprog{\tloc} = \qinst{\qlo}{br}{}\ e\ \bbc\ \offt\ \offf\) with the same conditions as above.
    \item iff \(\sprog{\sloc} = \qinst{\qs}{call}{}\ b\ \sloc_{\mathtt{foo}}\ \sloc_{\mathtt{foo'}}\)
    then \(\qinst{\qlo}{call}{}\ b\ \tloc_{\mathtt{foo|foo'}}\)
    with \(\sloc_{\mathtt{foo}} \corr \tloc_{\mathtt{foo|foo'}}\) and \(\sloc_{\mathtt{foo'}} \corr \tloc_{\mathtt{foo|foo'}} + 1\);
    \item for any other instruction, \(\sprog{\sloc} = \tprog{\tloc}\).
  \end{itemize}
\end{lemma}
\begin{proof}
  The proof follows from the definition of \(\transform(\sprog{})\).
\end{proof}

\subsection{Correctness}\label{app:correctness}
This section sketches the proof of \cref{thm:correctness}:
\correctness*

First, we define a relation \(\ctxstack \corrctx{} \retstack\), relating stack of program locations \(\retstack\) in a program \(\tprog{}\) to a stack of \libra{} contexts \(\ctxstack\). Intuitively, \(\corrctx{}\) holds iff contexts on the context stack \(\ctxstack\) are valid for program locations on the return stack \(\retstack\).
\begin{definition}[\(\protect\corrctx{}\)]\label{def:corrctx} For a \libra{} program \(\tprog{}\), \(\corrctx{}\) is defined as follows:
\begin{align*}
  (\bbc,\off) &\corrctx{} \loc &&
    \text{if in \tprog{}, } \loc \text{ points to offset } \off \text{ in a slice of size } \bbc\\
  \ctxstack \cdot (\bbc,\off)  &\corrctx{} \retstack \cdot \loc &&
    \text{if } (\bbc,\off) \corrctx{} \loc \text{ and } \ctxstack \corrctx{} \retstack
\end{align*}
\end{definition}

Second, we define an strong correspondence relation (\(\simulates\)) between source and target configurations. In addition to \(\libraeq{\sprog{}}\) (\cref{def:libraeq}), this stronger relation also requires correspondence between return stacks, and that the \libra{} stack is correct \wrt{} the return stack:
\begin{definition}[\(\sconf{} \protect\simulates{} \tconf{}\)]\label{def:libraeqstrong}
  A source configuration \(\sconf{} = \conf{\srcu{\mem},\srcu{\reg}, \srcu{\pc}, \srcu{\retstack}}\) for a program \(\sprog{}\) is strongly related to a target configuration \(\tconf{} = \conf{\trgb{\mem}, \trgb{\reg}, \trgb{\pc}, \trgb{\retstack}, \trgb{\ctxstack}}\) for a program \(\tprog{}\), denoted \(\sconf{} \simulates{} \tconf{}\), if and only if the following holds:
  \begin{itemize}
    \item \(\sconf{} \libraeq{\sprog{}} \tconf{}\)
    \item \(\srcu{\retstack} \corrloc{\sprog{}} \trgb{\retstack}\),
    \item \(\trgb{\ctxstack} \corrctx{} \tretstack \cdot \tpc\)
  \end{itemize}
  where (as defined before) \(\corrloc{\sprog{}}\) relates program locations in the source program to their corresponding location in the target program. Abusing notation, we lift \(\corrloc{\sprog{}}\) to return stacks.
\end{definition}

This strong correspondence relation will be our induction invariant to prove \cref{thm:correctness}.
The proof of \cref{thm:correctness} follows directly from the following lemma:
\begin{lemma}[Correctness-lemma]\label{thm:correctness-lemma}
  For all \(\sasm\) program \(\sprog{}\), number of steps \(n\), and initial source and target
  configurations \(\sconf{}\) and \(\tconf{}\) such that
  \(\tconf{} \simulates{} \sconf{}\):
  \begin{gather*}
    \text{if } %
    \sevalconf{\sconf{}}{\sconf{}'}{}{n}
    \text{ then }
    \tevalconf{\tconf{}}{\tconf{}'}{}{n}
    \text{ and } \sconf{}' \simulates{} \tconf{}'
  \end{gather*}
  where \(\seval{}\) is parameterized by \(\sprog{}\) and \(\teval{}\) is parameterized by \(\transform(\sprog{})\).
\end{lemma}

\begin{proof}
  The proof goes by induction on the evaluation relation in the source configuration. The base case trivially follows from the definition of initial configurations and \(\tconf{} \simulates \sconf{}\).

  For the inductive case, we assume that the source configuration makes a step and, as an induction hypothesis (IH), that configurations at step \(n-1\) (\(\sconf{}\) and \(\tconf{}\)) are related by \(\simulates{}\):
  \begin{gather}
    \sevalconf{\sconf{}}{\sconf{}'}{}{} \tag{Hstep}\label{corr:Hstep} \\
    \sconf{} \simulates{} \tconf{} \tag{IH}\label{corr:IH}
  \end{gather}
  We then show that the target configuration also makes a step such that the final configurations are still related by \(\simulates{}\):
  \begin{equation}\tag{Goal}\label{corr:G}
     \tevalconf{\tconf{}}{\tconf{}'}{}{} \text{ and }\sconf{}' \simulates{} \tconf{}'
  \end{equation}

  The proof goes by case analysis on the evaluated instruction. %
  In particular, for each rule in the source semantics, \cref{lemma:t-corr} gives us the instruction evaluated in the target, which, in turn, defines the rule that can be applied in the target semantics. For each rule, \(\tevalconf{\tconf{}}{\tconf{}'}{}{}\) always trivially follows from \ref{corr:Hstep}. Hence, so to show \ref{corr:G}, we only detail \(\sconf{}' \simulates{} \tconf{}'\).

  \proofcase{Rule \textsc{pc-update}.} In this rule, the memory, register map and program counters are modified.

  (Memory \& Registers) By \ref{corr:IH} and \cref{lemma:t-corr}, we know that both source and target configurations evaluate the same instruction.
  Additionally, from \ref{corr:IH}, the register and memory and the same in both configurations.
  Therefore, the final memories and register maps are also equal in the source and target.

  (Program counter) Let \(\bbc\), \(\off\) be the \libra{} context on top of \(\trgb{\ctxstack}\). From the semantics, we get \(\spc' = \spc + 1\) in the source, and \(\tpc' = \tpc + \bbc\) in the target. We need to show that \(\spc' \corr \tpc'\) and that \((\bbc,\off) \corrctx{} \tpc'\).

  From \ref{corr:IH} and \cref{def:corrctx}, we get that
  \((\bbc,\off) \corrctx{} \tpc{}\), meaning that, in \(\tprog{}\), \(\tpc\) points to a slice of size \(\bbc\) at offset \(\off\).
  Note that, because we are in the middle of a basic block, the next slice is adjacent to the current slice and has the same size.
  Hence, \(\tpc'\) point to the next slice (of size \(\bbc\)), at offset \(\off\), which concludes \((\bbc,\off) \corrctx{} \tpc'\).
  Additionally, from \(\transform(\sprog{})\), the next slice at offset \(\off\) also corresponds to \(\spc'\), which concludes \(\spc' \corr \tpc'\).

  \proofcase{Regular branch true.} (Case false is analogous.) In this rule, only the program counter is modified.
  By \ref{corr:IH} and \cref{lemma:t-corr}, we know that both source and target configurations are respectively evaluating \inst{br}{\(e\ \sloc_{t}\ \sloc_{f}\)} and either
  \(\inst{br}{}\ e\ \tloc_{t}\ \tloc_{f}\), with \(\sloc_{t} \corr \tloc_{t}\) and \(\sloc_{f} \corr \tloc_{f}\), or \qinst{\qlo}{br}{\(e\ \bbc'\ \offt\ \offf\)}.
  We focus here on the case \(\inst{br}{}\ e\ \tloc_{t}\ \tloc_{f}\) (see next proof case for \qinst{\qlo}{br}{}).
  We also know from \ref{corr:IH}, that \(e\) evaluates to the same value in source and target.

  (Program counter) From the semantics, we have that the final program counters are \(\spc' = \sloc_{t}\) and \(\tpc' = \tloc_{t}\). From \(\sloc_{t} \corr \tloc_{t}\), we can conclude \(\spc' \corr{} \tpc'\).

  Additionally, we need to show that \((\bbc, \off) \corrctx{} \tpc'\). From \cref{lemma:t-corr}, we know that \(\spc\) is not in a secret-dependent region, nor in a secret-dependent function.
  By definition of \(\transform\), \(\spc{} \corr \tpc{}\), and \((\bbc, \off) \corrctx{} \tpc\), it follows that \((\bbc, \off) = (1,0)\). From \cref{hyp:sese,hyp:cfi}, we know that \(\spc'\) is also not in a secret-dependent region, nor in a secret-dependent function. By definition of \(\transform{}\), it follows that \(\tpc'\) corresponds to a slice of size 1 (at offset 0), which concludes \((\bbc, \off) \corrctx{} \tpc'\).

  \proofcase{Rule \textsc{lob-true}.} (Case \textsc{lob-false} is analogous.) In this rule, the program counters and libra context stack are modified.
  By \ref{corr:IH} and \cref{lemma:t-corr}, we know that both source and target configurations are respectively evaluating
  \qinst{$\{\varepsilon,\qs\}$}{br}{\(e\ \sloc_{t}\ \sloc_{f}\)} and \qinst{\qlo}{br}{\(e\ \bbc'\ \offt\ \offf\)} instructions.
  We also know from \ref{corr:IH}, that \(e\) evaluates to the same value in source and target.

  (Program counter) From the semantics, we have that the final program counters are \(\spc' = \offt\) and \(\tpc' = \loc + \offt\) with \(\loc = \func{next\_slice}(\tpc, \bbc, \curOff)\) is the address of the next slice. Because \qinst{\qlo}{br}{} is a terminating instruction, this next slice is located in a different basic block \(B\). By definition of \(\transform{}\), we also know that \(B\) is adjacent to the current basic block, hence from the correctness of the \libra{} stack (\ref{corr:IH}), the start address of \(B\) is \(\loc\).
  Finally, from \cref{lemma:t-corr}, we have that \(\offt + \loc \corr \spc'\), which concludes \(\spc' \corr \tpc'\).

  (\libra{} context) The \libra{} context on top of the stack is updated to \((\bbc',\offt)\). We need to show that \((\bbc',\offt) \corrctx{} \tpc'\), \ie{} that in \tprog{}, \(\tpc'\) points to offset \(\offt\) in a slice of size \(\bbc'\).
  This follows directly from \cref{lemma:t-corr}.

  \proofcase{Regular calls.} In this rule, the program counter and return stack are modified.

  (Program counter) By \ref{corr:IH} and \cref{lemma:t-corr}, we know that source and target configurations respectively evaluate
  \(\inst{call}{}\ \sloc_{\mathtt{foo}}\) and \(\inst{call}{}\ \tloc_{\mathtt{foo}}\), with \(\sloc_{\mathtt{foo}} \corr \tloc_{\mathtt{foo}}\). From the semantics, we have \(\spc' = \sloc_{\mathtt{foo'}}\) and \(\tpc' = \tloc_{\mathtt{foo}}\). Hence, \(\spc' \corr \tpc'\).

  (Return stack)
  In the source, the return address \(\spc + 1\) is added to the stack. In the target, \(\tpc + \bbc\) is added to the stack.
  To show that \(\sretstack' \corr \tretstack'\), it suffices to show \(\spc + 1 \corr \tpc + \bbc\). See proof for the rule \textsc{pc-update}, case (program counter).

  (Libra context) The \libra{} context \((1, 0)\) is added on top of the \libra{} context stack. We need to show:
  \begin{equation*}
    \ctxstack \cdot (\bbc,\off) \cdot (1,0) \corrctx{} \tretstack \cdot \tpc + \bbc \cdot \tpc'
  \end{equation*}
  From \cref{def:corrctx}, this amounts to showing:
  \begin{gather}
    (1,0) \corrctx{} \tpc' \label{corr:call:1}\\
    (\bbc,\off) \corrctx{} \tpc + \bbc \label{corr:call:2}\\
    (\bbc,\off) \corrctx{} \tretstack \label{corr:call:3}
  \end{gather}
  \begin{itemize}
    \item[(\ref{corr:call:1})] \(\tpc'\) points to the function \(\mathtt{foo}\), because \(\mathtt{foo}\) is not at folded function (by definition of \(\transform{}\)), we have \((1,0) \corrctx{} \tpc'\).
    \item[(\ref{corr:call:2})] From \ref{corr:IH}, we have \((\bbc,\off) \corrctx{} \tpc\). Moreover, because we are in the middle of a basic block, \(\tpc + \bbc\) still points to a slice (of size \bbc), at offset \off, which concludes \((\bbc,\off) \corrctx{} \tpc + \bbc\).
    \item[(\ref{corr:call:3})] \((\bbc,\off) \corrctx{} \tretstack\) directly follows from \ref{corr:IH}.
  \end{itemize}

  \proofcase{Rule \textsc{lo-call} \(\bot\).} (Case \(\top\) is analogous.) In this rule, the program counters, return stack, and \libra{} contexts are modified.

  (Program counter) By \ref{corr:IH} and \cref{lemma:t-corr}, we know that both source and target configurations evaluate \qinst{\qs}{call}{\(\bot\ \sloc_{\mathtt{foo}}\ \sloc_{\mathtt{foo'}}\)} and \qinst{\qlo}{call}{\(\bot\ \tloc_{\mathtt{foo|foo'}}\)} instructions, respectively, with
  \(\sloc_{\mathtt{foo'}} \corr \tloc_{\mathtt{foo|foo'}} + 1\);
  From the semantics, we have \(\spc' = \sloc_{\mathtt{foo'}}\) and \(\tpc' = \tloc_{\mathtt{foo|foo'}} + 1\). Hence, \(\spc' \corr \tpc'\).

  (Return stack) See the proof for regular calls, case return stack.

  (Libra context) The \libra{} context \((2, 1)\) is added on top of the \libra{} context stack. We need to show:
  \begin{equation*}
    \ctxstack \cdot (\bbc,\off) \cdot (2,1) \corrctx{} \tretstack \cdot \tpc + \bbc \cdot \tpc'
  \end{equation*}
  The proof is similar to the proof for regular calls. The only difference is the subcase \ref{corr:call:1}, which becomes
  \((2,1) \corrctx{} \tpc'\). %
  Here, \(\tpc'\) points to the first slice of the folded function \(\mathtt{foo|foo'}\) at offset \(1\) (dummy part). Additionally (from \(\transform(\sprog{})\)) the first slice of a folded function is always of size 2. Hence \((2,1) \corrctx{} \tpc'\).

  \proofcase{Rule \textsc{ret}.} In this rule, the program counters, return stack, and \libra{} contexts are modified.
  By \ref{corr:IH} and \cref{lemma:t-corr}, we know that both source and target configurations evaluate \inst{}{ret}{}.

  (Program counter) The program counter are updated to point to the address on
  the top of the return stack. From \ref{corr:IH}, these addresses are related by \(\corr\). This concludes \(\spc' \corr \tpc'\).

  (Return stack) The return address is popped from the return stack, in both source and target. Hence, \(\sretstack' \corr \tretstack'\) directly follow from \ref{corr:IH}.

  (Libra context) The previous \libra{} context is restored. We must show that it is a correct context for the next configuration (\ie{} \(\tctxstack' \corrctx{} \tretstack' \cdot \tpc'\)). Notice that under
  \cref{hyp:cfi}, \(\tctxstack'\) cannot be empty. Hence,
  \(\tctxstack' \corrctx{} \tretstack' \cdot \tpc'\) directly follows from \ref{corr:IH}.
\end{proof}

The following lemma follows from the fact that source program are safe (\cref{hyp:safe}), and that source and target semantics are deterministic.
\begin{lemma}\label{lemma:src-step}
  For all \(\sasm\) program \(\sprog{}\), number of steps \(n\), and source and target
  configurations \(\sconf{}\) and \(\tconf{}\) such that
  \(\tconf{} \simulates{} \sconf{}\):
  \begin{gather*}
    \text{if } %
    \tevalconf{\tconf{}}{\tconf{}'}{}{n}
    \text{ then }
    \sevalconf{\sconf{}}{\sconf{}'}{}{n}
    \text{ and } \sconf{}' \simulates{} \tconf{}'
  \end{gather*}
  where \(\seval{}\) is parameterized by \(\sprog{}\) and \(\teval{}\) is parameterized by  \(\transform(\sprog{})\).
\end{lemma}

\subsection{Security}\label{app:security}
We recall our hypothesis about the leakage model:
\begin{hypothesis}\label{hyp:leakage} We assume that:
  \begin{itemize}
    \item Secure branches and level-offset branches do not leak their outcome: \qinst{\(\{\qlo,\qs\}\)}{br}{\(c\ \_\)} \(\in \IuSafe{}\),
    \item Secure calls and level-offset calls do not reveal whether the original function or the dummy function is actually executed: \qinst{\(\{\qlo,\qs\}\)}{call}{\_} \(\in \IuSafe{}\),
    \item Normal branches leak their outcome: \inst{br}{\(c\ \_\)} \(\in \IuUnsafe{}\),
    \item Normal calls leak their target: \inst{call}{\(\loc\)} \(\in \IuUnsafe{}\),
    \item Control-flow-altering instructions belong in a distinct leakage class from each other and from non-control-flow-altering instructions.
  \end{itemize}
\end{hypothesis}

\noindent{}
We now sketch the proof of \cref{thm:security}:
\security*

\begin{proof}

  Let \(\sprog{}\) be a \(\oni{\obsfuncweak}\) program, %
  we need to prove that \(\transform(\sprog{})\) is \(\oni{\obsfuncstrong}\).  %
  To this end, we assume a pair of executions of \(\transform(\sprog{})\):
  \begin{equation*}
    \tconf{}_{0} \teval{\tobs_{0}} \dots \teval{\tobs_{n-1}} \tconf{}_{n} \qquad
    \tconf{}_{0}' \teval{\tobs_{0}'} \dots \teval{\tobs_{n-1}'} \tconf{}_{n}'
  \end{equation*}
  starting from low-equivalent configurations, \ie{} \(\tconf{}_{0} \indist{} \tconf{}_{0}'\). We need to prove that both executions produce the same leakage, \ie{} \(\tobs_{i} = \tobs_{i}'\) for all \(0 \leq i < n\). By \cref{lemma:src-step}, there exist corresponding source executions:
  \begin{equation*}
    \sconf{}_{0} \seval{\sobs_{0}} \dots \seval{\sobs_{n-1}} \sconf{}_{n} \qquad
    \sconf{}_{0}' \seval{\sobs_{0}'} \dots \seval{\sobs_{n-1}'} \sconf{}_{n}'
  \end{equation*}
  such that \(\sconf{}_{0} \indist{} \sconf{}_{0}'\) and for all \(0 \leq i \leq n\), \(\sconf_{i} \simulates \tconf_{i}\) and \(\sconf_{i}' \simulates \tconf_{i}'\).
  From \(\oni{\obsfuncweak}(\sprog{})\), we have \(\sobs_{i} = \sobs_{i}'\) for all \(0 \leq i < n\).

By the definition of source and target leakages, for \(0 \leq i < n\), we have:
  \begin{gather*}
    \tobs_{i} = \func{slice\_addr}(\tpc_{i}, \trgb{\curOff}_{i}) \cdot \obsfuncweak(\conf{\tmem_{i}, \treg_{i}, \tpc_{i}})\\
    \tobs_{i}' = \func{slice\_addr}(\tpc_{i}', \trgb{\curOff}_{i}') \cdot \obsfuncweak(\conf{\tmem_{i}', \treg_{i}', \tpc_{i}'})\\
    \sobs_{i} = \obsfuncweak(\conf{\smem_{i}, \sreg_{i}, \spc_{i}}) =
    \sobs_{i}' = \obsfuncweak(\conf{\smem_{i}', \sreg_{i}', \spc_{i}'})
  \end{gather*}
  To show that the target leakages are equal, we consider the balanceable part of the target leakage (\ie{} \(\obsfuncweak(\conf{\tmem_{i}, \treg_{i}, \tpc_{i}})\) and \(\obsfuncweak(\conf{\tmem_{i}', \treg_{i}', \tpc_{i}'})\)), and the unbalanceable part of the target leakage (\ie{} \(\func{slice\_addr}(\tpc_{i}, \trgb{\curOff}_{i})\) and \(\func{slice\_addr}(\tpc_{i}', \trgb{\curOff}_{i}')\)) separately.

  \proofcase{Balanceable leakage.} This case follows a similar pattern as compilation proofs for constant-resource transformation~\cite{DBLP:conf/csfw/BartheBHP21}.
  In particular, \(\transform{}\) is a leakage preserving
  transformation for the balanceable part of the leakage. This means that the target execution produces exactly the same observation as its corresponding source execution.
  Hence, \(\obsfuncweak(\conf{\tmem_{i}, \treg_{i}, \tpc_{i}}) = \sobs_{i}\) and \(\obsfuncweak(\conf{\tmem_{i}', \treg_{i}', \tpc_{i}'}) = \sobs_{i}'\).
  From \(\sobs_{i} = \sobs_{i}'\), we can conclude that target leakages are identical, which concludes our goal.

  \proofcase{Unbalanceable leakage.} This case follows a similar pattern as compilation proofs for constant-time preservation. In particular, we need to find a lockstep CT-simulation \wrt{} \(\simulates\), which relates source states and target states at each execution step and which yields preservation of ONI. In compilation proofs for constant-time preservation, this CT-simulation is usually simple program counter equivalence. %
  In our case, however, program counters of low-equivalent executions can diverge. Instead, our CT-simulation relates executions pointing to the \emph{same program slice}.

  Indeed, from \cref{hyp:leakage,hyp:sese}, we get that for all \(0 \leq i \leq n\), the control-flows of \(\sconf_{i}\) and \(\sconf_{i}'\) belong to the same slice. More precisely:
  \begin{itemize}
    \item \(\sconf_{i}\) is inside a secret dependent region iff \(\sconf_{i}'\) is inside a secret-dependent region,
    \item when not in a secret-dependent region, \(\spc_{i} = \spc_{i}'\), and
    \item when in a secret-dependent region, \(\spc_{i}\) and \(\spc_{i}'\) belong
    to the same level and are at the same offset in their respective basic
    blocks.
  \end{itemize}
  We remark that \(\transform{}\) folds secret-dependent regions such that instructions from the same level that are at the same offset in their respective basic blocks are placed in the same slice.

  Hence, \(\sconf_{i} \simulates \tconf_{i}\) and \(\sconf_{i}' \simulates \tconf_{i}'\), we get that \(\func{slice\_addr}(\tpc_{i}, \trgb{\curOff}_{i})\) and \(\func{slice\_addr}(\tpc_{i}', \trgb{\curOff}_{i}')\) correctly return the base address of their current slices and that these slice are the same in both target executions, which concludes our goal and the proof.

\end{proof}

\fi

\end{document}